\documentclass[11pt,a4paper]{article}

\usepackage{amsmath,amssymb,amsthm}
\usepackage{geometry}
\usepackage{hyperref}
\usepackage{enumitem}

\newtheorem{theorem}{Theorem}[section]
\newtheorem{lemma}[theorem]{Lemma}
\newtheorem{proposition}[theorem]{Proposition}
\newtheorem{corollary}[theorem]{Corollary}

\theoremstyle{definition}
\newtheorem{definition}[theorem]{Definition}
\newtheorem{remark}[theorem]{Remark}
\newtheorem{observation}[theorem]{Observation}

\newcommand{\MQ}{\mathrm{MQ}}
\newcommand{\FF}{\mathbb{F}}
\newcommand{\VV}{\mathcal{V}}
\newcommand{\Fix}{\mathrm{Fix}}
\newcommand{\RM}{\mathrm{RM}}
\newcommand{\cl}{\mathrm{cl}}
\newcommand{\ev}{\mathrm{ev}}
\newcommand{\GRM}{\mathrm{GRM}}
\newcommand{\M}{\overline{M}}
\newcommand{\dist}{\mathtt{d}}
\DeclareMathOperator{\supp}{supp}
\DeclareMathOperator{\wt}{w}

\DeclareMathOperator{\Rank}{Rank}

\title{\textbf{A Note on Binary Quadratic Systems and Their Relation to Complexity Theory}}
\author{%
  Gabriele Radici\thanks{Department of Mathematics, University of Trento, Italy.
  \texttt{gabriele.radici@unitn.it}}
  \and
  Massimiliano Sala\thanks{Department of Mathematics, University of Trento, Italy.
  \texttt{maxsalacodes@gmail.com}}%
}

\begin{document}

\maketitle

\begin{abstract}
Deciding whether a system of multivariate quadratic equations over
$\FF_2$ has a solution is a classical NP-complete problem, and remains so for
square systems, having as many equations as variables. 
The hardness of this problem is one of the
cornerstones of modern post-quantum cryptography.
Let $\MQ_0(n)$ and $\MQ_1(n)$ denote the sets of square quadratic
systems in $n$ variables having no solutions and exactly
one solution, respectively. $\bigcup_{n\geq 2} \MQ_0(n)$ is a coNP-complete language, while
$\bigcup_{n\geq 2} \MQ_1(n)$ lies in DP. 
It is known that
$\lim_{n\to \infty} |\MQ_1(n)|/|\MQ_0(n)|=1$.
Here we prove the explicit bounds - for a finite $n$ -
\[
|\MQ_0(n)|<|\MQ_1(n)|
\le
\left(1+\frac{1}{2^n-1}\right)|\MQ_0(n)|.
\]
We find of interest that the two previous inequalities are equivalent to the existence of injections
$
\MQ_0(n)\hookrightarrow\MQ_1(n)
$
(whose images omit at most a $2^{-n}$ fraction of $\MQ_1(n)$).
While whether a total injection can be given explicitly (or efficiently computed or inverted) remains open, a polynomial partial injection over a dense set of $\MQ_0(n)$ into $\MQ_1(n)$ is constructed in a follow-up work \cite{radicisala}.

More generally, let $Q_d$ be the space of polynomial functions
$(\FF_2)^n\to\FF_2$ of degree at most $d$, and let $\alpha_k$
count square systems in $(Q_d)^n$ having exactly $k$ solutions. Then
\[
\alpha_0<\alpha_1
\le
\left(1+\frac{1}{2^n-1}\right)\alpha_0, \qquad 2\le d\le n \,.
\]
The proof combines matroid- and coding-theoretic methods. We interpret
$(\FF_2)^n$ as the ground set of the evaluation matroid of $Q_d$,
express $\alpha_0$ and $\alpha_1$ through characteristic polynomials,
and use a Whitney-type sign-reversing involution to show that the only
terms that can push $\alpha_1-\alpha_0$ below $\alpha_1/2^n$ come from the
elements of a matroid port. These are identified with
minimal-support words of the Reed--Muller code
$\RM(n-d-1,n) = \RM(d,n)^\perp$; the required estimate then follows
from the MacWilliams identity, the minimum-distance bound $2^{d+1}$,
and the even-weight structure of the code.
\end{abstract}

\section{Introduction}\label{sec:intro}

In this paper, we present results of an algebraic, geometric and combinatorial nature, which may open new paths to results in complexity theory. These results are also closely connected to post-quantum cryptography.
The reader unfamiliar with complexity theory will find in Section~\ref{sec:complexity} a collection of notions and statements useful to understand the connection with our algebraic results.

From now on, $q$ denotes a prime power and 
$\FF_q$ denotes the finite field with $q$ elements, in some papers denoted by
$\mathrm{GF}(q)$.
Deciding whether a system of multivariate quadratic equations over
$\FF_2$ has at least one solution is a classical NP-complete problem
\cite{fraenkelyesha,gareyjohnson}. As shown in Section~\ref{sec:complexity}, NP-completeness persists when one restricts to
\emph{square} systems, that is, to systems with as many equations as variables.
Beyond its complexity-theoretic status, the hardness of this problem (in the more general finite-field setting) is the main hardness assumption underlying
multivariate post-quantum cryptography. For example, of the nine candidates advanced by NIST to
the third round of its additional digital signature process in May 2026, four are variants of the Unbalanced Oil and Vinegar scheme \cite{uov}: UOV itself, MAYO \cite{mayo}, QR-UOV \cite{qruov} and SNOVA \cite{snova}, whose security rests on the difficulty of solving systems of quadratic equations over $\FF_q$
\cite{nistir8610}. 

For $n\ge 2$ and $0 \leq k\leq 2^n$, let $\MQ_k(n)$ denote the set of square quadratic systems in $n$ variables over $\FF_2$ having $k$ solutions. In particular,
$\MQ_0(n)$ and $\MQ_1(n)$ denote the sets of square quadratic systems in $n$ variables over $\FF_2$ having no solutions and exactly one solution, respectively. The language $\bigcup_{n\ge 2}\MQ_0(n)$ is coNP-complete, whereas $\bigcup_{n\ge 2}\MQ_1(n)$ belongs to DP and is not known to lie in either NP or coNP. Moreover, $\bigcup_{n\ge 2}\MQ_1(n)$ is coNP-hard and, under randomised reductions, NP-hard. For the unique satisfiability problem, these are results of Blass and Gurevich \cite{blassgurevich} and of Valiant and Vazirani \cite{valiantvazirani}; their adaptation to square quadratic systems is detailed in Section~\ref{sec:complexity}.

The two families have remarkably close cardinalities. It is known \cite{fuscobach} that
\[
\lim_{n\to\infty} \frac{|\MQ_1(n)|}{|\MQ_0(n)|}=1.
\]
Here we prove the explicit finite-$n$ bounds
\begin{equation}\label{eq:intro-mq}
|\MQ_0(n)| < |\MQ_1(n)| \le \left(1+\frac{1}{2^n-1}\right)|\MQ_0(n)|.
\end{equation}
Therefore, with a relative excess of order only $2^{-n}$, a uniformly-random (binary) square quadratic system is \emph{more likely to have exactly one solution than none}.\\

The two inequalities in
\eqref{eq:intro-mq} are equivalent to the existence of an
injection
\[
\MQ_0(n)\hookrightarrow\MQ_1(n),\quad \mbox{for any } n\geq 2\,,
\]
whose image omits at most a $2^{-n}$ fraction of $\MQ_1(n)$.
Whether a total injection can be given explicitly, or computed/inverted in polynomial time, remains open: an injection of that kind would relate a coNP-complete language to a language in DP. However, in a follow-up work \cite{radicisala}, a partial injection $\tau$ is constructed that maps a subset $\mathcal{D} \subseteq \MQ_0(n)$ of density greater than $74.5\%$ into $\MQ_1(n)$ in $\Theta(n^3)$ bit operations. \\
While we make no
complexity-theoretic claim in this paper,
in Section~\ref{sec:complexity} we discuss more.
\\

From now on we focus on the binary field. Our result for quadratic systems is a special case of the general
statement that we prove in this paper. Let
$n\ge 2$, $V=(\FF_2)^n$, $N=|V|=2^n$ and
\[
Q:=\FF_2[x_1,\dots,x_n]/\langle x_1^2+x_1,\dots,x_n^2+x_n\rangle
\]
for the ring of polynomial functions $V\to\FF_2$ (also called \emph{binary polynomials}). 

For $1\le d\le n$,
let $Q_d\subseteq Q$ be the subspace of functions of degree at most
$d$, of dimension
$
D=D(n,d)=\sum_{j=0}^{d}\binom n j
$.\\
A \emph{square system of degree $\le d$} is an $n$-tuple
$F\in (Q_d)^{\,n}$, with \emph{variety} $\VV(F)\subseteq V$, that is,
\[
F=(f_1,\dots,f_n)\in (Q_d)^{\,n}, \qquad 
\VV(F):=\{v\in V \mid f_1(v)=\dots=f_n(v)=0\}.
\]
For $0\le k\le N$, set
\[
\alpha_k=\alpha_k(n,d)
:=
\bigl|\{F\in (Q_d)^{\,n}:|\VV(F)|=k\}\bigr|.
\]
Thus, for $d=2$ our notation provides
\[
|\MQ_k(n)|=\alpha_k(n,2).
\]

\newpage

Our main result is the following.
\begin{theorem}[restated with proof as Theorem~\ref{thm:main}]\label{thm:intro}
Let $n\ge 2$ and $2\le d\le n$. Then
\[
  \Bigl(1-\frac{2}{5N^{2}}\Bigr)\frac{\alpha_1}{N}
  \;\le\;\alpha_1-\alpha_0\;\le\;\frac{\alpha_1}{N}.
\]
In particular, the following inequality holds: \begin{equation*}
  \alpha_1 \ > \ \alpha_0 
\end{equation*}
\end{theorem}

Theorem~\ref{thm:intro} specialises to
\eqref{eq:intro-mq}, since
\[
\alpha_1-\alpha_0\le\frac{\alpha_1}{N}
\quad\Longleftrightarrow\quad
\alpha_1\le
\left(1+\frac1{N-1}\right)\alpha_0\,.
\]
The equality in its upper bound occurs if and only if
$n=d$.
\\

Two boundary cases of Theorem~\ref{thm:intro} are classical: $d=1$ and $d=n$, as follows. 
\\
For $\boxed{d=1}$ the number of solutions of a random
square linear system over $\FF_2$ has the distribution
governed by the rank of a random matrix \cite[Chapter~3]{kolchin},
and there the inequality actually goes the \underline{other} way: $\alpha_0>\alpha_1$ for
every $n\ge3$ (Theorem~\ref{thm:affine}). \\
As a consequence, the hypothesis $d\ge2$ in
Theorem~\ref{thm:intro} cannot be dropped.
\\
For
$\boxed{d=n}$, $Q_n$ is the space of all functions $V\to\FF_2$, the number of
solutions of a uniformly random square system is exactly binomial with
parameters $N$ and $1/N$ \cite[Theorem~1]{jain}, and the upper bound holds with
equality. For $d<n$ the number of solutions is no longer binomial, because the events ``$F$ vanishes at $v$'', $v\in V$, are independent only for $d=n$ \cite[Corollary~3.6]{jain}; this is what makes the case $d<n$ non-trivial.

We briefly describe the proof. Section~\ref{sec:prelim} collects the facts on matroids that we use. 

In Section~\ref{sec:matroid}, the point set $V$ is made
into a matroid $\M$ by declaring $S\subseteq V$
independent when the evaluation vectors of its points on $Q_d$ are
linearly independent. Inclusion--exclusion then gives
\[
\alpha_0=\chi_{\M}(N),
\qquad
\alpha_1=N\chi_{\M/v}(N),
\]
where $\chi$ denotes the characteristic polynomial and $v\in V$ is
arbitrary.

In Section~\ref{sec:reduction}, a sign-reversing involution in the
spirit of Whitney's broken-circuit theorem, adapted to the port of $\M$
at $v$, shows that the only terms that can push
$\alpha_1-\alpha_0$ below $\alpha_1/N$ come from the minimal sets of
the port of odd cardinality. More precisely, if $\mu_j$ denotes the
number of minimal sets of the port of cardinality $j$, we obtain
\[
\alpha_1-\alpha_0
\ge
\frac{\alpha_1}{N}
-
\sum_{j\ {\rm odd}}\mu_jN^{D+1-j}.
\]

In Section~\ref{subsec:rm-codes}, these minimal sets are identified
with the minimal-support words through $v$ of the Reed--Muller code
\[
\RM(n-d-1,n)=\RM(d,n)^\perp,
\]
so that the $\mu_j$ become weight-enumerator coefficients. 
Their minimum (Hamming) distance is denoted by $\dist$ to distinguish it from the degree $d$. 

In
Section~\ref{sec:rm-bound} the resulting error term is controlled through the
MacWilliams identity: the distance $\dist = 2^{d+1}$ implies agreement with
the corresponding binomial model through the first moments, while the
even-weight structure gives the required sign for the upper bound, and an
estimate on the moments makes the error term small enough.

Section~\ref{sec:mainthm} assembles these ingredients into a proof of
Theorem~\ref{thm:intro}, simultaneously showing that its degree hypothesis $d\ge2$ cannot be
relaxed.

Finally, in Section~\ref{sec:complexity} we comment on the
relation between the research carried out in this paper and complexity
theory, while in Section~\ref{sec:open} we conclude our paper with some open problems.

\section{Preliminaries on matroids}\label{sec:prelim}

Matroids, introduced by Whitney \cite{whitney35}, abstract the notion of linear independence; we refer to \cite{welsh,oxley} for the basic theory. We use the rank-function axiomatisation.

\begin{definition}\label{def:matroid-rank}
A \emph{matroid} $M$ is a pair $(E,r)$, where $E$ is a finite set, the \emph{ground set}, and $r:2^E\to\mathbb Z_{\ge0}$ satisfies, for all $A,B\subseteq E$:
\begin{enumerate}[label=\textnormal{(R\arabic*)}]
    \item $0\le r(A)\le|A|$;
    \item if $A\subseteq B$ then $r(A)\le r(B)$;
    \item $r(A\cup B)+r(A\cap B)\le r(A)+r(B)$.
\end{enumerate}
The rank of $M$ is, by definition, $r(M):=r(E)$.
\end{definition}

Matroids are also usually defined by axioms on the independent sets; the two
axiomatisations are equivalent, in the following sense.

\begin{theorem}[\cite{whitney35,welsh}]\label{thm:rank-equivalence}
Call $\mathcal I\subseteq2^E$ an \emph{independence system on} E if
\begin{itemize}
    \item $\emptyset\in\mathcal I$; \item subsets of members of $\mathcal I$ are in $\mathcal I$; \item if $I,J\in\mathcal I$ and $|I|<|J|$ there is $e\in J\setminus I$ with $I\cup\{e\}\in\mathcal I$.
\end{itemize}
The assignments
\[
r\ \longmapsto\ \mathcal I_r:=\{I\subseteq E: r(I)=|I|\},
\qquad
\mathcal I\ \longmapsto\ r_{\mathcal I}(A):=\max\{|I| : I\subseteq A,\ I\in\mathcal I\}
\]
are mutually inverse bijections between the set of functions $r:2^E\to\mathbb Z_{\ge0}$
satisfying \textnormal{(R1)}--\textnormal{(R3)} and the set of independence systems on $E$.
\end{theorem}

\begin{remark}
This explicit correspondence is necessary to uniquely recover $r$ from $\mathcal I_r$. 
\end{remark}

From now, $M$ always denotes a matroid $(E,r)$. \\

We recall the following standard concepts in matroid theory \cite{oxley,welsh,whitney35}.

\begin{definition}\label{def:basic}
Let $M = (E, r)$ be a matroid with rank function $r$.
\begin{enumerate}
    \item An \emph{independent} set is an element of $\mathcal{I}_r$, while a set in $2^E \setminus \mathcal{I}_r$ is \emph{dependent}.
    \item A \emph{basis} $B$ is a maximal set in $\mathcal{I}_r$, the set of all bases is denoted $\mathcal{B}(M)$. 
    \item A \emph{circuit} $C$ is a minimal dependent set, the set of all circuits is denoted $\mathcal C(M)$.
    \item An element $e\in E$ is a \emph{loop} if $r(\{e\})=0$. 
    \item An element $e\in E$ is a \emph{coloop} if $r(E\setminus\{e\})=r(E)-1$. 
    \item Two distinct non-loop elements $e,f \in E$ are \emph{parallel} if $r(\{e,f\})=1$.
    \item The map $\cl:2^E\to2^E$, $\cl(A):=\{e\in E: r(A\cup\{e\})=r(A)\}$ is the \emph{closure operator}. \\
    For $A \subseteq E$, $\cl(A)$ is called the \emph{closure of} $A$.
    \item For $T\subseteq E$, the \emph{deletion} matroid $M \setminus T$ is the matroid ($E\setminus T,r_{M\setminus T})$, where for $A \subseteq E \setminus T$, $r_{M\setminus T}(A):=r(A)$.
    \item For $T\subseteq E$, the \emph{contraction} matroid $M / T$ is the matroid $(E\setminus T,r_{M / T})$, where for $A \subseteq E \setminus T$, $r_{M/ T}(A):=r(A \cup T) - r(T)$.
    \item The \emph{characteristic polynomial} of $M$ is the polynomial (in $\mathbb{Z}[\lambda]$) \[\chi_M(\lambda)=\sum_{A\subseteq E}(-1)^{|A|}\lambda^{r(M)-r(A)}.\]
\end{enumerate}
\end{definition}

\begin{remark}\label{rem:facts}
The following are standard facts \cite{welsh,oxley}.
\begin{enumerate}[label=\textnormal{(\alph*)}]
    \item All bases have cardinality $r(M)$.
    \item A subset $C \subseteq E$ is a circuit iff $r(C) = |C| - 1$ and $r(C \setminus \{e\}) = |C| - 1$ for all $e \in C$.
    \item The closure operator $\cl: 2^E \to 2^E$ satisfies the following properties for any set $S \subseteq E$: 
    \begin{itemize}
        \item Extensivity: $S \subseteq \cl(S)$;
        \item Monotonicity: if $S \subseteq S'$ then $\cl(S) \subseteq \cl(S')$;
        \item Idempotency: $\cl(\cl(S)) = \cl(S)$.
    \end{itemize}
    \item If $I$ is independent and $e\in\cl(I)\setminus I$, then $I\cup\{e\}$ contains exactly one circuit, and that circuit contains $e$. It is called the \emph{fundamental circuit} of $e$ with respect to $I$, we denote it as $C(e,I)$.
    \item $r_{M\setminus T}$ and $r_{M/T}$ are rank functions. We recall the short proof for $r_{M/T}$.  We have that, since $A \cap T = \emptyset$ and by the submodularity of $r$, $r(A \cup T) - r(T) \le |A|$. $(R1),(R2)$ are immediate. To prove submodularity of $r_{M/T}$, one notices that $r(A \cup B \cup T) + r((A \cap B) \cup T) = r((A \cup T) \cup (B \cup T)) + r((A \cup T) \cap (B \cup T)) \le r(A \cup T) + r(B \cup T)$.
    \item The characteristic polynomial of $M\setminus T$ is $\chi_{M\setminus T} = \sum_{S \subseteq E \setminus T}(-1)^{|S|}\lambda^{r(M \setminus T) - r(S)}$;
    \item The characteristic polynomial of $M/T$ is $\chi_{M/T} = \sum_{T \subseteq S \subseteq E}(-1)^{|S| - |T|}\lambda^{r(M) - r(S)}$ \\ $ = \sum_{S \subseteq E \setminus T}(-1)^{|S|}\lambda^{r(M) - r(S \cup T)}$
\end{enumerate}
\end{remark}

\begin{remark}\label{rem:chi-split}
For any matroid $M=(E,r)$ and any $e\in E$, the following identity holds
\begin{equation}\label{eq:chi-split}
\chi_M(\lambda)=\sum_{S\subseteq E\setminus\{e\}}(-1)^{|S|}\lambda^{r(M)-r(S)}-\sum_{S\subseteq E\setminus\{e\}}(-1)^{|S|}\lambda^{r(M)-r(S\cup\{e\})}.
\end{equation}
\end{remark}

\eqref{eq:chi-split} is used in proving a useful identity involving characteristic polynomials of deletion and contraction matroids:

\begin{proposition}[\cite{welsh,oxley}]\label{identity}
Let $M=(E,r)$ be a matroid and let $e\in E$.
If $e$ is not a coloop, then
    \[
    \chi_M(\lambda)=\chi_{M\setminus e}(\lambda)-\chi_{M/e}(\lambda).
    \]
    In particular:
    \begin{itemize}
        \item if $e$ is a loop, then $\chi_M(\lambda)=0$ and $\chi_{M\setminus e}(\lambda)=\chi_{M/e}(\lambda)$;
        \item if $e$ has a parallel element, then $\chi_{M/e}(\lambda)=0$ and $\chi_M(\lambda)=\chi_{M\setminus e}(\lambda)$.
    \end{itemize}
\end{proposition}

\section{Boolean polynomial systems as matroids} \label{sec:matroid}

The point set $V=(\FF_2)^n$ carries a matroid structure in which a set of points is
independent when the corresponding evaluation functionals on $Q_d$ are linearly
independent. The purpose of this section is to set that structure up and to record
the primary consequence that the rest of the paper uses: the numbers $\alpha_0$ and
$\alpha_1$ of systems with no solution and with exactly one solution are values of
the characteristic polynomial of this matroid and of one of its contractions.

\subsection{Rank function}

We recall our notation: $\FF_2$ is the field with two elements, $V = (\FF_2)^n$ is the $n$-dimensional vector space over $\FF_2$, and $N = 2^n = |V|$.
Set $Q := \FF_2[x_1, \dots, x_n] / \langle x_1^2 + x_1, \dots, x_n^2+x_n \rangle$ and, for $1 \le d \le n$, let $Q_d$ denote the space of polynomials in $Q$ up to degree $d$. The dimension of $Q_d$ is 
\[
D = \sum_{j=0}^d \binom{n}{j}.
\]
Let $\textrm{Mon}_d$ be the set of monomials of $Q_d$. Since $\mathrm{Mon}_d$ is a basis for $Q_d$, $Q_d$ is isomorphic to $(\FF_2)^D$. For a monomial $m = x_{i_1}\dots x_{i_d}$, we call $\supp(m) = \{i_1,\dots,i_d\}$ the \emph{support} of $m$.
We order monomials in $\mathrm{Mon}_d$ with respect to the \textit{graded lexicographic} ordering. Namely, $m < m'$ if and only if: \begin{itemize}
    \item $|\supp(m)| < |\supp(m')|$;
    \item $|\supp(m)|=|\supp(m')|$ and $\min(\supp(m) \triangle \supp(m')) \in \supp(m)$,
\end{itemize}
where for two sets $A,B$, $\triangle$ is the symmetric difference $A \triangle B = (A \cup B) \setminus (A \cap B)$. \\

Define the evaluation map \[ \ev : V \to (\FF_2)^D, \qquad  \ev(v) := (m(v) : m \in \mathrm{Mon}_d).\] Here $(\FF_2)^D$ carries two roles, identified throughout by the standard scalar product: it is the space of coefficient vectors $\overline{f}$ of the polynomials of $Q_d$, and, through $\overline{f}\mapsto\overline{f}\cdot\ev(v)$, it is also the space of linear functionals on $Q_d$. 

\begin{remark}
    $\ev$ is injective: since $\ev(v) = \big(1, v, (m(v))_{ \deg m \ge 2}\big)$, $\ev(v)=\ev(w)$ implies $v = w$.
    
\end{remark}

Extending $\ev$ to the power set of $V$, 
$$\ev:2^V \to 2^{(\FF_2)^D}, \qquad \ev(S) = \{\ev(v) : v \in S\}\,,
$$
we can define the set function:
\[
\overline{r}(S) := \Rank(\ev(S)), \quad \text{for } S \subseteq V.
\]

\begin{remark}
    By construction, also the extended $\ev:2^V \to 2^{(\FF_2)^D}$ is clearly injective.
\end{remark}

\begin{lemma}\label{obs:boolean-matroid}
The function $S \mapsto \overline{r}(S)$ is a matroid rank function on ground set $V$.
\end{lemma}

\begin{proof}
We directly verify axioms (R1)--(R3) of Definition~\ref{def:matroid-rank} for $\overline{r}(S)$:

\textbf{(R1)}: $\ev(S)$ is a collection of $|S|$ vectors in $(\FF_2)^D$. The dimension of $\mathrm{span}(\ev(S))$ cannot exceed the number of vectors $|S|$ nor can it be negative, so $0 \le \overline{r}(S) \le |S|$.

\textbf{(R2)}: If $A \subseteq B \subseteq V$, then $\ev(A) \subseteq \ev(B)$. Taking linear spans yields $\mathrm{span}(\ev(A)) \subseteq \mathrm{span}(\ev(B))$, so $\dim \mathrm{span}(\ev(A)) \le \dim \mathrm{span}(\ev(B))$, i.e., $\overline{r}(A) \le \overline{r}(B)$.

\textbf{(R3)}: Let $A, B \subseteq V$. Set $U = \mathrm{span}(\ev(A))$ and $W = \mathrm{span}(\ev(B))$. Then $U + W = \mathrm{span}(\ev(A \cup B))$. Since $\ev(A \cap B) \subseteq \ev(A) \cap \ev(B) \subseteq U \cap W$, we have $\mathrm{span}(\ev(A \cap B)) \subseteq U \cap W$, so $\overline{r}(A \cap B) \le \dim(U \cap W)$. Applying the standard dimension formula for vector subspaces:
\begin{align*}
\overline{r}(A \cup B) + \overline{r}(A \cap B) &= \dim(U + W) + \overline{r}(A \cap B) \\
&\le \dim(U + W) + \dim(U \cap W) \\
&= \dim U + \dim W = \overline{r}(A) + \overline{r}(B). \qedhere
\end{align*}
\end{proof} 

\begin{corollary}
    $\M = (V,\overline{r})$ is a matroid. We call it the \emph{evaluation matroid} (of degree $d$).
\end{corollary}

\begin{remark}\label{rem:rv_eq_D}
    $\ev(V)$ spans the whole $(\FF_2)^D$ (since no non-zero polynomial vanishes on all $V$), hence $\overline{r}(V) = \dim(\FF_2)^D = D$. 
\end{remark}

Having established that $\overline{r}$ is a rank function on the power set of $V$, for compactness we will write $r_S := \overline{r}(S)$, $S \subseteq V$. 

\subsection{Translation of matroid concepts}

Having proven that $(V, \overline{r})$ is a matroid, we specialise general definitions to this setting for better understanding.

\begin{enumerate}
    \item \emph{Independent Sets $\mathcal{I}$}: Point sets $I \subseteq V$ whose monomial evaluation vectors $\ev(I)$ are linearly independent in $(\FF_2)^D$.
    \item \emph{Bases $\mathcal{B}(\M)$}: Maximal point sets in $V$ whose evaluation vectors form a vector space basis for $\mathrm{span}(\ev(V)) = (\FF_2)^D$.
    \item \emph{Circuits $\mathcal{C}(\M)$}: Minimal point sets $C \subseteq V$ exhibiting a non-trivial linear dependency $\sum_{v \in C} c_v \ev(v) = 0$ over $\FF_2$. As shown in Section \ref{subsec:rm-codes}, these correspond to minimal-support non-zero codewords in the dual Reed--Muller code $\RM(d, n)^\perp = \RM(n - d - 1, n)$.
    \item \emph{Deletion $\M \setminus T$}: For any $T \subseteq V$, restricting evaluation to the subset of points $V \setminus T$.
    \item \emph{Contraction $\M/T$}: For any $T \subseteq V$, projecting the linear functionals space $(\FF_2)^D$ modulo $\mathrm{span}(\ev(T))$.
    \item \emph{Characteristic Polynomial:} $\chi_{\M}(\lambda) = \sum_{S \subseteq V} (-1)^{|S|} \lambda^{D - r_S}$.
\end{enumerate}

Regarding loops, coloops and parallel elements in $\M$, we have the following result: 

\begin{lemma}\label{lem:coloop}
$\M$ has no loops, and no two points are parallel. If $d\le n-1$, then $\M$ has no coloops; if $d=n$ every point is a coloop.
\end{lemma}

\begin{proof}
For any $v \in V$, $\ev(v)$ has first coordinate $1(v)=1$, so it is non-zero, hence it cannot have zero rank. \\ $v,w \in V$ are parallel iff $\ev(v)=\ev(w)$, impossible since $\ev$ is injective. \\
For $v \in V$ denote $s(v) = \{i : v_i = 1\}$, then for $U \subseteq V$ it holds that $\sum_{v \in U}m(v) = |\{v \in U \mid \supp(m) \subseteq s(v)\}| \mod 2$. \\
Fix a vector $v \in V$. If $d\le n-1$ then for any $m \in \mathrm{Mon}_d$, $\sum_{w \in V}m(w) = 0$, hence $\ev(v)$ is linearly dependent with $\ev(V \setminus \{v\})$ and so $r_{V\setminus\{v\}}=r_V=D$. 
On the other hand, if $d=n$, for every $U \subseteq V \setminus \{v\}$ take a vector $w$ of maximum weight in $U \cup \{v\}$, then the monomial $m_w$ with $\supp(m) = s(w)$ is such that $\sum_{u \in U}m_w(u) + m(v) = 1$. Hence $\ev(v)$ is linearly independent with $\ev(V \setminus \{v\})$ and so $r_{V\setminus\{v\}}=r_V - 1= N - 1$.  
\end{proof}

\subsection{Counting systems with 0 and 1 solutions}

For the rest of the paper we consider $d \ge 2$, unless otherwise stated. \\

For $S\subseteq V$ let
\[
H_S:=\{f\in Q_d: f(v)=0\ \forall v\in S\} \subseteq (\FF_2)^D.
\]
Let $\overline{f} \in (\FF_2)^D$ represent the vector of coefficients of the polynomial $f = \sum f_i m_i$. By construction $f(v) = \overline{f}\cdot \ev(v)$ where $\cdot$ is the standard scalar product of $(\FF_2)^D$. \\

This implies that $H_S = \ev(S)^{\perp}$, so $|H_S| = 2^{D - r_S}$. \\

Recalling that $\mathcal{V}(F)$ is the variety of $F \in (Q_d)^n$ and that $N=2^n$, we define
\begin{equation}
\varphi(S) := \big| \{F \in (Q_d)^n \mid S \subseteq \mathcal{V}(F) \} \big| =|(H_S)^n|,
\end{equation}
hence we have \[ \varphi(S)
= |H_S|^n = 2^{n(D - r_S)} = N^{D - r_S}. \] Indeed the space of systems vanishing (at least) on $S$ is $(H_S)^n$, since a system vanishes on $S$ if and only if all its polynomials vanish on $S$. \\

We recall that for $k = 0, \dots, 2^n$, $\alpha_k$ denotes  $\big|\{F \in (Q_d)^n \mid |\mathcal{V}(F)| = k\}\big|$, i.e. the number of square systems with $k$ solutions.
Recalling that $r_V = D$ (Remark~\ref{rem:rv_eq_D}), we show that

\begin{theorem}\label{thm:boolean-matroid-link}
In the evaluation matroid $\M$, for any $v \in V$:
\[
\alpha_0 = \chi_{\M}(N) \quad \text{and} \quad \alpha_1 = \chi_{\M/v}(N)N.
\]
\end{theorem}

\begin{proof}
Applying inclusion-exclusion directly to the complement of the arrangement of all the $(H_v)^n$ inside $(Q_d)^n \simeq (\FF_2)^{nD}$ we get
\[
\alpha_0 = \left| (\FF_2)^{nD} \setminus \bigcup_{v \in V} H_v^n \right| = \sum_{S \subseteq V} (-1)^{|S|} \left| \bigcap_{v \in S} H_v^n \right| = \sum_{S \subseteq V} (-1)^{|S|} \left|H_S^n\right| = \sum_{S \subseteq V} (-1)^{|S|} N^{D - r_S} = \chi_{\M}(N).
\]
For the second identity, fix $v \in V$. Systems whose unique solution is $v$ are the points of the set $H_v^n \setminus \bigcup_{w \neq v} (H_v^n \cap H_w^n)$. By inclusion-exclusion, their cardinality is:
\begin{align*}
    \left| (H_v)^n \setminus \bigcup_{w \neq v} ((H_v)^n \cap (H_w)^n)  \right| = \sum_{S \subseteq V \setminus \{v\}} (-1)^{|S|} \left| \bigcap_{w \in S} ((H_v)^n \cap (H_w)^n) \right|  \\ =  \sum_{S \subseteq V \setminus \{v\}} (-1)^{|S|} \varphi(S \cup \{v\}) = \sum_{S \subseteq V \setminus \{v\}} (-1)^{|S|}N^{D - r_{S \cup \{v\}}} = \chi_{\M/v}(N). 
\end{align*}

Furthermore, the affine transformation group of $V$ acts transitively on its $N$ points, preserves degree, induces an invertible linear transformation of the coefficient space $(\FF_2)^D$, and thus it preserves the count for every point, yielding $\alpha_1 = \chi_{\M/v}(N)N$.
\end{proof}

\begin{remark}
In Theorem~\ref{thm:boolean-matroid-link}, the absence of loops and parallel elements in $\M$ (Lemma~\ref{lem:coloop}) ensures that neither count vanishes trivially. 
That would be a contradiction, since we can always construct systems with either zero or one solutions.
\end{remark}

\section{A sign-reversing involution and the matroid inequality} \label{sec:reduction}

In this section $M=(E,r)$ is an arbitrary matroid and $e\in E$. We reduce $\lambda\chi_{M/e}(\lambda)-\chi_M(\lambda)$ to sums over two families of \emph{independent} sets, and then bound it from below.
Although we have in mind applications to our specific matroid $\M$, the results are more clear when presented in their general form that we present in this section.
The specialization to $\M$ is presented in the next Section.

The three steps are as follows. In \S\ref{subsec:ports} the difference is split
along the \emph{port} of $M$ at $e$: the sets $S$ with $e\in\cl(S)$ contribute one
sum, the remaining ones another. Neither sum is alternating in a useful way, because
the ranks $r(S)$ vary with $S$. In \S\ref{subsec:involutions} we remove that
obstacle with a sign-reversing involution in the spirit of Whitney's
broken-circuit theorem: on each family it pairs off all but the sets carrying no
\emph{active} element, and paired sets have equal rank and opposite sign, so they
cancel. What survives is indexed by \emph{independent} sets only, on which
$r(S)=|S|$, so both sums become polynomials in $\lambda$ with explicit integer
coefficients $\sigma_k^{(0)}$ and $\sigma_k^{(1)}$ counting fixed sets of each size. In
\S\ref{subsec:fixed} we bound those coefficients: each family is closed under
passing to subsets, up to the members of the port itself, which gives
useful inequalities for both $\sigma_k^{(0)}$ and $\sigma_k^{(1)}$ with an error
term $\mu_{k+1}$ counting the elements of the port of size $k+1$. Pairing consecutive terms then makes the
first sum non-negative and the second small, which is
Theorem~\ref{thm:matroid-inequality}.

\subsection{Matroid ports}\label{subsec:ports}

The port of a matroid was first introduced by Lehman in 1964 in a work on Shannon's Switching Game \cite{port}.
\begin{definition}
    For $e \in E$, the $e$-\emph{port} of $M$, or the \emph{port of M at} $e$, is the set of circuits containing $e$, without $e$. \[\mathcal{P}_e := \{ C \setminus \{e\} \mid C \in \mathcal{C}(M), \ e \in C\}\]
    Elements of $\mathcal{P}_e$ are called \emph{minimals of} $e$.  
\end{definition}

We define the upfilter generated by it and the associated polynomial.
\begin{definition}\label{def:minimals}
For $e \in E$, define the $e$-\emph{port upfilter} as \[\mathcal{T}_e := \{ S \subseteq E \setminus \{e\} \mid \exists  \ C \in \mathcal{P}_e \text{ s.t. } C \subseteq S\} \] and the $e$-\emph{port polynomial}
\[
W_e(\lambda) = \sum_{S \in \mathcal{T}_e} (-1)^{|S|} \lambda^{r(M) - r(S)}.
\]
\end{definition}

\begin{remark} \label{obs:active-circuit} 
    The upfilter $\mathcal{T}_e$ can be characterised via closure as $\mathcal{T}_e = \{ S \subseteq E \setminus \{e\} \mid e \in \mathrm{cl}(S)\}$. Indeed, for $e \notin S$, $e \in \mathrm{cl}(S)$ iff there exists a circuit $C \subseteq S \cup \{e\}$, with $e \in C$. Note also that $\mathcal{P}_e \subset \mathcal{T}_e$ and that the elements of $\mathcal{P}_e$ are minimals w.r.t the inclusion inside $\mathcal{T}_e$ (as per Definition~\ref{def:minimals} ).
\end{remark}

\begin{remark}\label{rem:we-edge}
    If $e$ is a loop then $\mathcal{P}_e = \{\emptyset\}$, $\mathcal{T}_e = E \setminus {e}$ and $W_e = \chi_{M\setminus e} = \chi_{M/e}$.\\
    If $e$ is a coloop then $\mathcal{P}_e=\mathcal{T}_e=\emptyset$ and $W_e=0$.
\end{remark}

\begin{proposition}\label{prop:We}
Let $M$ be a matroid on ground set $E$. Then for every $e \in E$,
\begin{equation*}
    \chi_{M/e}(\lambda)(\lambda - 1)= \chi_M(\lambda) + (\lambda - 1) W_e(\lambda),
\end{equation*} equivalently
\begin{equation}\label{eq:We}
  \chi_{M/e}(\lambda)\lambda - \chi_M(\lambda) = \chi_{M/e}(\lambda) + (\lambda - 1) W_e(\lambda).  
\end{equation}
\end{proposition}

\begin{proof}
We prove the second equality applying the identity \eqref{eq:chi-split}. \\
For any $S \subseteq E \setminus \{e\}$, consider the term differences between $ \chi_{M/e}(\lambda)\lambda$ and $\chi_{M \setminus e}(\lambda)$, i.e. $(-1)^{|S|}(\lambda^{r(M)-r(S \cup \{e\}) + 1} - \lambda^{r(M)-r(S)}) $. If $e \notin \mathrm{cl}(S)$, adjoining $e$ increases rank by 1, yielding matching terms. If $e \in \mathrm{cl}(S)$, the difference evaluates to $(\lambda - 1)(-1)^{|S|} \lambda^{r(M) - r(S)}$.
\begin{align*}
&\chi_{M/e}(\lambda)\lambda - \chi_M(\lambda) \\ &= \lambda \sum_{S \subseteq E \setminus \{e\}} (-1)^{|S|} \lambda^{r(M) - r(S \cup \{e\})} - \sum_{S \subseteq E \setminus \{e\}} (-1)^{|S|} \lambda^{r(M) - r(S)} + \sum_{S' \subseteq E \setminus \{e\}} (-1)^{|S'|} \lambda^{r(M) - r(S' \cup \{e\})} \\
&= (\lambda - 1) \sum_{\substack{S \subseteq E \setminus \{e\} \\ e \in \mathrm{cl}(S)}} (-1)^{|S|} \lambda^{r(M) - r(S)} + \sum_{S' \subseteq E \setminus \{e\}} (-1)^{|S'|} \lambda^{r(M) - r(S' \cup \{e\})} \\
&= \chi_{M/e}(\lambda) + (\lambda - 1) W_e(\lambda). 
\qedhere
\end{align*}
\end{proof}

\begin{corollary}\label{cor:we-edge}
Let $M=(E,r)$ be a matroid and $e\in E$.
\begin{enumerate}[label=\textnormal{(\alph*)}]
    \item If $e$ is a coloop, then  $\chi_M(\lambda) = (\lambda - 1)\chi_{M/e}(\lambda)$.
    \item If $e$ has a parallel element, then $\chi_M(\lambda) = (1-\lambda)W_e(\lambda)$.
\end{enumerate}
\end{corollary}

\begin{proof}
The first identity follows by Remark~\ref{rem:we-edge}, since $W_e=0$, while the second from Propositon~\ref{identity}, since $\chi_{M/e}=0$.
\end{proof}

Fixing $v \in V$ this can be translated in the evaluation matroid as

\begin{corollary} \label{cor:We-alpha}
    \begin{equation*}
        \alpha_1 - \alpha_0 = \frac{\alpha_1}{N} + (N - 1)\sum_{S \in \mathcal{T}_v}(-1)^{|S|}\varphi(S)
    \end{equation*}
\end{corollary}

\begin{proof}
    Recalling that $\alpha_0 = \chi_{\M}(N)$ and $\alpha_1 = \chi_{\M/v}(N)N$ (hence $\chi_{\M/v}(N) = \frac{\alpha_1}{N}$), and evaluating the formula \eqref{eq:We} at $\lambda = N$, we get the result.
\end{proof}

\subsection{Involutions}\label{subsec:involutions}

The sign of $W_e(\lambda)$ is hard to read off because the ranks of the sets in $\mathcal{T}_e$ vary. We use a sign-reversing involution, in the spirit of Whitney's broken-circuit theorem \cite{whitney32} and its bijective proof by Blass and Sagan \cite{blasssagan}, to rewrite each of the sums in \eqref{eq:We} as a sum over independent sets.

Fix a total order $e_1<e_2<\dots$ on $E$. For $S\subseteq E$ and $p\in E$ let $S_{>p}:=\{s\in S: s>p\}$. Put $E_0 := E$ and $E_j:=E\setminus\{e_1,\dots,e_j\}$, for $j \ge 1$.

\begin{definition}\label{def:active}
Let $j\ge0$ and $S\subseteq E_j$. An element $p\in E_j$ is $j$-\emph{active} for $S$ if $p\in\cl(S_{>p})$. \\
If $S$ has a $j$-active element, its $j$-\emph{pivot} is $p_j(S):=\min\{p\in E_j: p\ j\text{-active for }S\}$ and \\ $\iota_j(S):=S\,\triangle\,\{p_j(S)\}$; otherwise $\iota_j(S):=S$. 
\end{definition}

\begin{definition}
For $T \subseteq 2^{E_j}$ put $\Fix_j(T) := \{S \in T \mid \iota_j(S) = S\}$.
\end{definition}

\begin{observation}
    $p \in \mathrm{cl}(S_{>p})$ if and only if there is a circuit $C \subseteq S_{>p} \cup \{p\}$ with $p \in C$. Moreover $\min C = p$ since $C \setminus \{p\} \subseteq S_{>p}$ and hence exceeds $p$. For any $e_j$ such that $p > e_j$, $C$ is therefore a circuit of $M \setminus \{e_1, \dots, e_j\}$. 
\end{observation}

\begin{lemma}\label{lem:involution-full}
Let $j\ge0$. Then:
\begin{enumerate}[label=\textnormal{(\alph*)}]
\item[(i)] the map $\iota_j:2^{E_j}\to2^{E_j}$ is an involution, and every non-fixed
orbit consists of two sets $\{S,\iota_j(S)\}$ with $|\iota_j(S)|=|S|\pm1$;
\item[(ii)] if $S\notin\Fix_j(2^{E_j})$, then $r(S\cup U)=r(\iota_j(S)\cup U)$ for every $U\subseteq E$. In
particular $r(S)=r(\iota_j(S))$;
\item[(iii)] any $S\in\Fix_j(2^{E_j})$ is independent, so $r(S)=|S|$.
\end{enumerate}
\end{lemma}

\begin{proof}
(i): We prove that for any $S \subseteq E_j$, $p_j(S) = p_j(\iota_j(S))$. This will give that $\iota_j(\iota_j(S)) = \iota_j(S \triangle \{p_j(S)\}) = (S \triangle \{p_j(S)\}) \triangle \{p_j(S)\} = S$, i.e. $\iota_j$ is an involution. \\
Let $p = p_j(S)$ and $S' = \iota_j(S) = S \triangle \{p\}$. 
One immediately observes that $S'_{>p} = S_{>p}$, so $p$ remains active for $S'$. \\
We claim that it is minimal with this property. For $q < p$ we have $S_{>p} \subseteq S_{>q}$, hence $\mathrm{cl}(S_{>p}) \subseteq \mathrm{cl}(S_{>q})$, and since $p \in \mathrm{cl}(S_{>p})$ we get $p \in \mathrm{cl}(S_{>q})$; therefore
\[
\mathrm{cl}(S_{>q} \cup \{p\}) = \mathrm{cl}(S_{>q}) \quad \text{and} \quad \mathrm{cl}(S_{>q} \setminus \{p\}) \subseteq \mathrm{cl}(S_{>q}).
\]
The equality (when $p$ is added) shows adjoining $p$ to $S_{>q}$ is closure-neutral, so it cannot bring $q$ into the closure; the inclusion (when $p$ is removed) shows deleting $p$ can only shrink the closure. In both directions, $q \in \mathrm{cl}(S'_{>q})$ would force $q \in \mathrm{cl}(S_{>q})$, equivalently $q$ active for $S$, contradicting the minimality of $p$. Hence $p$ is the least active element of $S'$, so $p(S') = p$ and $\iota_j$ is an involution.

(ii): Let $A = S \setminus \{p\}$, so $\{S, \iota_j(S)\} = \{A, A \cup \{p\}\}$. The active pivot satisfies $p \in \cl(S_{>p}) \subseteq \cl(A)$. For $U \subseteq E$, monotonicity implies $p \in \cl(A \cup U)$, which forces $r(A \cup \{p\} \cup U) = r(A \cup U)$. Thus $r(S \cup U) = r(\iota_j(S) \cup U)$. Taking $U = \emptyset$ gives $r(S) = r(\iota_j(S))$.

(iii): By definition of $\iota_j$, $S$ is fixed iff it has no
$j$-active element. Such an $S$ is independent: if $S$ contained a circuit $C$, then
$p:=\min C$ would satisfy $C\setminus\{p\}\subseteq S_{>p}$, hence
$p\in\cl(S_{>p})$, so $p$ would be $j$-active for $S$, a contradiction. Hence
$r(S)=|S|$.
\end{proof}

\begin{corollary}[Restricted Whitney Theorem]\label{cor:whitney}\ \\
If $T\subseteq2^{E_j}$ satisfies $\iota_j(T)=T$ and $U\subseteq E$, then
\[
\sum_{S\in T}(-1)^{|S|}\lambda^{r(M)-r(S \, \cup \, U )}=\sum_{S\in\Fix_j(T)}(-1)^{|S|}\lambda^{r(M)-r(S \, \cup \, U)} .
\]
In particular,
\[
\chi_M(\lambda)=\sum_{S\in\Fix_0(2^E)}(-1)^{|S|}\lambda^{r(M)-|S|}.
\]
\end{corollary}

\begin{proof}
By Lemma~\ref{lem:involution-full}(i),(ii) the non-fixed members of $T$ split into
pairs $\{S,\iota_j(S)\}$ with $|\iota_j(S)|=|S|\pm1$ and
$r(S\cup U)=r(\iota_j(S)\cup U)$, whose contributions cancel. For the second
formula use Lemma~\ref{lem:involution-full}(iii) to replace $r(S)$ by $|S|$ on the
fixed sets.
\end{proof}

We now prove the $\iota_j$-invariance of two subsets of $M$, that will appear in the general formula in Corollary \ref{cor:pointed-reduction-full}.

\begin{lemma}\label{lem:invariance-Te-full}
Let $e := e_1$. The following hold:
\begin{itemize}
    \item[1.] $\mathcal{T}_e$ is $\iota_1$-invariant and for $S \in \Fix_1(\mathcal{T}_e)$, $r(S) = r(S \cup \{e\})=|S|$;
    \item[2.] $\overline{\mathcal{T}_e} := \{S \in 2^E \setminus \mathcal{T}_e \mid e \notin S\}$ is $\iota_0$-invariant and for $S \in \Fix_0(\overline{\mathcal{T}_e})$, $r(S)=|S|$.
\end{itemize}
\end{lemma}

\begin{proof}
(1): Let $\iota = \iota_1$, $S \in \mathcal{T}_e$. Assume whitout loss of generality that $S \ni p = p_1(S)$ and $A = S \setminus \{p\}$ (otherwise it is enough to exchange their role). The closure operator is monotone and idempotent, so $p \in \mathrm{cl}(A)$ forces
\[
\mathrm{cl}(A) \subseteq \mathrm{cl}(A \cup \{p\}) \subseteq \mathrm{cl}(\mathrm{cl}(A)) = \mathrm{cl}(A),
\]
that is, $\mathrm{cl}(A \cup \{p\}) = \mathrm{cl}(A)$. In particular $\mathrm{cl}(\iota(S)) = \mathrm{cl}(S) = \mathrm{cl}(A)$ and $r(\iota(S)) = r(S)$. Consequently membership of any fixed element in the closure is unaffected by the toggle:
\[
e \in \mathrm{cl}(S) \iff e \in \mathrm{cl}(\iota(S)).
\]
Hence $\iota$ maps $\mathcal{T}_e$ to itself. \\
The rank formula is given directly by the definition of $\mathcal{T}_e$ and Lemma \ref{lem:involution-full} (iii).\\

(2): Let $S\in\overline{\mathcal{T}_e}$. Since $e\notin S$ we have $S_{>e}=S$, so $e\in\cl(S_{>e})$ would mean $S\in \mathcal{T}_e$; thus $e$ is not $0$-active for $S$, $p_0(S)\ne e$ if it exists, and $\iota_0(S)\subseteq E\setminus\{e\}$. If $S$ is not fixed, $\cl(\iota_0(S))=\cl(S)\not\ni e$, so $\iota_0(S)\in\overline{\mathcal{T}_e}$. Fixed members are independent by Lemma~\ref{lem:involution-full} (iii), and $r(S\cup\{e\})=r(S)+1$ since $e\notin\cl(S)$.
\end{proof}

\begin{corollary}\label{cor:pointed-reduction-full}
If $e := e_1$ is neither a loop nor a coloop, then
\[
 \chi_{M/e}(\lambda)\lambda - \chi_M(\lambda) = \sum_{S \in \Fix_0(\overline{\mathcal{T}_e})} (-1)^{|S|} \lambda^{r(M) - r(S \cup \{e\})} + \lambda \sum_{S \in \Fix_1(\mathcal{T}_e)} (-1)^{|S|} \lambda^{r(M) - r(S)}.
\]
\end{corollary}

\begin{proof}
Both $\overline{\mathcal{T}_e}$ and $\mathcal{T}_e$ consist of subsets of $E\setminus\{e\}$, and
$2^{E\setminus\{e\}}=\overline{\mathcal{T}_e}\sqcup \mathcal{T}_e\,$. Splitting the contraction formula
of Remark~\ref{rem:facts} (g) accordingly, and using $r(S\cup\{e\})=r(S)$ for
$S\in \mathcal{T}_e$ (by definition $e\in\cl(S)$ there), gives
\[
\chi_{M/e}(\lambda)=\underbrace{\sum_{S\in\overline{\mathcal{T}_e}}
(-1)^{|S|}\lambda^{r(M)-r(S\cup\{e\})}}_{=:A(\lambda)}
\;+\;\underbrace{\sum_{S\in \mathcal{T}_e}
(-1)^{|S|}\lambda^{r(M)-r(S)}}_{=\,W_e(\lambda)} ,
\]
where we have defined a convenient polynomial $A(\lambda)$.
By Proposition~\ref{prop:We},
\[
\lambda\chi_{M/e}(\lambda)-\chi_M(\lambda)
=\chi_{M/e}(\lambda)+(\lambda-1)W_e(\lambda)
=\bigl(A(\lambda)+W_e(\lambda)\bigr)+(\lambda-1)W_e(\lambda)
=A(\lambda)+\lambda\,W_e(\lambda)\,.
\]
Finally, $\overline{\mathcal{T}_e}$ is $\iota_0$-invariant and $\mathcal{T}_e$ is $\iota_1$-invariant
by Lemma~\ref{lem:invariance-Te-full}, so Corollary~\ref{cor:whitney} applies to
each sum: with $(j,T,U)=(0,\overline{\mathcal{T}_e},\{e\})$ it replaces $A(\lambda)$ by the
sum over $\Fix_0(\overline{\mathcal{T}_e})$, and with $(j,T,U)=(1,\mathcal{T}_e,\varnothing)$ it
replaces $W_e(\lambda)$ by the sum over $\Fix_1(\mathcal{T}_e)$. This is the stated formula.
\end{proof}

The previous result translates in our evaluation matroid setting in the following way.
\begin{corollary}
    \[
    \alpha_1 - \alpha_0 = \sum_{S \in \Fix_0(\overline{\mathcal{T}_v})} (-1)^{|S|}\varphi(S \cup \{v\}) + N\sum_{S \in \Fix_1(\mathcal{T}_v)} (-1)^{|S|}\varphi(S).
    \]
\end{corollary}

\begin{proof}
    Thanks to Lemma~\ref{lem:coloop}, we can evaluate the formula of Corollary~\ref{cor:pointed-reduction-full} at $\lambda = N$ and observe that $\alpha_1 - \alpha_0 = N\chi_{\M/v}(N) - \chi_{\M}(N)$ and $N^{D-r_S} = \varphi(S)$. \end{proof}

\subsection{The structure of the fixed families and the main inequality}\label{subsec:fixed}

\begin{lemma}\label{lem:interval}
For every set $S \in \Fix_1(\mathcal{T}_e)$:
\begin{enumerate}[label=\textnormal{(\roman*)}]
    \item $S \cup \{e\}$ contains a unique circuit $C = C(e,S)$; moreover $e \in C$.
    \item The set $S^e := C \setminus \{e\}$ is the unique member of the $e$-port $\mathcal{P}_e$ that is a subset of $S$.
    \item For every $S^e \subseteq S' \subseteq S$, $S' \in \Fix_1(\mathcal{T}_e)$.
\end{enumerate}
\end{lemma}

\begin{proof}
$S$ is independent and $e\in\cl(S)\setminus S$, so Remark~\ref{rem:facts} (d) gives the unique circuit \\ $C := C(e,S)$ and $e \in C$. If $P\in \mathcal{P}_e$ and $P\subseteq S$, then $P\cup\{e\}$ is a circuit in $S\cup\{e\}$, hence equals $C$. Let $S^e\subseteq S'\subseteq S$. Then $e\in\cl(S^e)\subseteq\cl(S')$, so $S'\in \mathcal{T}_e$; and if $p\in E_1$ were $1$-active for $S'$ then $p\in\cl(S'_{>p})\subseteq\cl(S_{>p})$ would be $1$-active for $S$. So $S'\in\Fix_1(\mathcal{T}_e)$.
\end{proof}

For $k\ge0$ set
\[
\sigma_k^{(0)}:=|\Sigma_k^{(0)}|=\bigl|\{S\in\Fix_0(\overline{\mathcal{T}_e}):|S|=k\}\bigr|,\quad
\sigma_k^{(1)}:=|\Sigma_k^{(1)}| =\bigl|\{S\in\Fix_1(\mathcal{T}_e):|S|=k\}\bigr|.
\]
Then $\sigma_0^{(0)}=1$ ($\emptyset\in\Fix_0(\overline{\mathcal{T}_e})$ since $e$ is not a loop) and $\sigma_0^{(1)}=\sigma_1^{(1)}=0$ ($e$ is not a loop and has no parallel element).  

By Lemma ~\ref{lem:invariance-Te-full}, posing $\mathcal S_1(\lambda)=W_e(\lambda)$, $\mathcal S_0(\lambda)=\chi_{M/e}(\lambda)-W_e(\lambda)$, Corollary~\ref{cor:pointed-reduction-full} reads
\begin{equation}\label{eq:ab}
\lambda\chi_{M/e}(\lambda)-\chi_M(\lambda)=\mathcal S_0(\lambda)+\lambda\,\mathcal S_1(\lambda),
\end{equation}
where
\[
\mathcal S_0(\lambda):=\sum_{k\ge0}(-1)^k\sigma_k^{(0)}\lambda^{r(M)-k-1},\qquad
\mathcal S_1(\lambda):=\sum_{k\ge0}(-1)^k\sigma_k^{(1)}\lambda^{r(M)-k}.
\]

With the previous simplifications, the sum becomes more tractable and the underlying sets are more structured than the starting ones. 
In particular, thanks to Lemma~\ref{lem:invariance-Te-full}, the exponent of $\lambda$ no more depends on the rank of $S$, but only on its cardinality.

We need now to study the coefficients. In accordance with the notation above, for $k\ge 0$, let \[ 
\mu_k:=\bigl|\{P\in \mathcal{P}_e:|P|=k\}\bigr| \]

\begin{lemma}\label{lem:shadow}
For every $k\ge0$:
\begin{enumerate}[label=\textnormal{(\roman*)}]
\item $\sigma_{k+1}^{(0)}\le(|E|-1-k)\,\sigma_k^{(0)}$;
\item $\sigma_{k+1}^{(1)}\le(|E|-1-k)\,\sigma_k^{(1)}+\mu_{k+1}$.
\end{enumerate}
\end{lemma}

\begin{proof}
(i) We note that $\Fix_0(\overline{\mathcal{T}_e})$ is closed under taking subsets. Indeed, let $S \in \Fix_0(\overline{\mathcal{T}_e})$ and $S' \subseteq S$. A $0$-active element for $S'$, would be such also for $S$, directly leading to a contradiction. \\
Hence every set $S$ in $\Sigma_{k+1}^{(0)}$ is of the form $S' \cup \{p\}$, with $S' \in \Sigma_k^{(0)}, \ p \in E \setminus S'$, with $p \ne e$. This gives the upper bound.  \\
\medskip

(ii) Let $S \in \Sigma_{k+1}^{(1)}$. By Lemma~\ref{lem:interval}, $S$ contains a unique minimal port set $S^e \in \mathcal{P}_e$. Every such $S$ falls into one of two classes:
\begin{enumerate}[label=\textnormal{\arabic*.}]
    \item $S$ is non-minimal ($S^e \subsetneq S$);
    \item $S$ is a minimal of $e$ ($S = S^e \in \mathcal{P}_e$).
\end{enumerate}
For sets in the first class an analog reasoning as in (i) gives $(|E|-1-k)\,\sigma_k^{(1)}$. The sets in the second class are $\mu_{k+1}$ by definition. Summing the two gives (ii).
\end{proof}

Lemma~\ref{lem:shadow} is very useful for proving bounds in the summations. For future readability, let \[
\mathcal{M}(\lambda) = \sum_{j\ \mathrm{odd}}\mu_j\,\lambda^{\,r(M)+1-j}
\] 

\begin{theorem}\label{thm:matroid-inequality}
Let $M=(E,r)$ be a matroid.
Let $e\in E$ be neither a loop nor a coloop and let it have no parallel element. Then, with the notation above,
\begin{enumerate}[label=\textnormal{(\alph*)}]
\item $\mathcal S_0(|E|)\ge 0$;
\item $|E|\,\mathcal S_1(|E|)\ge- \mathcal{M}(|E|)$;
\item if $\mathcal{P}_e \neq \emptyset$ and every member of $\mathcal{P}_e$ has odd size, then $\mathcal S_1(|E|)<0$.
\end{enumerate}
\end{theorem}

\begin{proof}
Write $\lambda=|E|$. \\
(a) Group the terms of $\mathcal S_0$ in pairs $(k,k+1)$ with $k$ even:
$$\mathcal S_0 = \sum_{k\ \mathrm{even}}|E|^{\,r(M)-k-2}\bigl(|E|\sigma_k^{(0)}-\sigma_{k+1}^{(0)}\bigr)\,.$$  By Lemma~\ref{lem:shadow}(i), $\sigma_{k+1}^{(0)}\le |E| \sigma_k^{(0)}$, hence $\mathcal{S}_0$ is a sum of non-negative terms, non-negative itself. 

\medskip
(b) Group the terms of $\mathcal S_1$ in pairs $(k,k+1)$ with $k$ even: 
$$
\mathcal S_1 =   \sum_{k\ \mathrm{even}}(-1)^k|E|^{\,r(M)-k-1}\bigl(|E|\sigma_k^{(1)}-\sigma_{k+1}^{(1)}\bigr)\,.$$ $|E|^{r(M)-k-1}(|E| \sigma_k^{(1)}-\sigma_{k+1}^{(1)})\ge-\mu_{k+1}|E|^{r(M)-k-1}$ by Lemma~\ref{lem:shadow} (ii). Reindexing with $j = k + 1$ (which is always odd) and multiplying by $|E|$ one gets the inequality. 

\medskip
(c) Group $\mathcal S_1$ in pairs $(k,k+1)$ with $k$ odd \footnote{recall that $\sigma_0^{(1)}=0$, so we can ignore $k=0$.}; then $\mu_{k+1}=0$, so $|E| \sigma_k^{(1)} - \sigma_{k+1}^{(1)} \ge 0$ and thus each term $-|E|^{r(M)-k-1}(|E| \sigma_k^{(1)}-\sigma_{k+1}^{(1)})$ is $\le0$. 
The strictness follows by the fact that any element of minimum size $\overline{k}$ in the $e$-port is fixed by $\iota_1$, hence $\sigma^{(1)}_{\overline{k}}\ge1$. So the pair $(\overline{k},\overline{k}+1)$ yields a negative term, thus $\mathcal{S}_1(|E|)<0$.

\end{proof}

\begin{corollary}\label{cor:inequality}
    Under the same hypotesis of Theorem~\ref{thm:matroid-inequality}, the following inequality holds. \[
\mathcal S_0(|E|)-\mathcal{M}(|E|) \le |E|\,\chi_{M/e}(|E|)-\chi_M(|E|)\, .
\]
Moreover, if $\mathcal{P}_e \neq \emptyset$ and every member of $\mathcal{P}_e$ has odd size,
\[ |E|\chi_{M/e}(|E|)-\chi_M(|E|)<\chi_{M/e}(|E|) < \mathcal S_0(|E|).\]
\end{corollary}

\begin{proof}
    Both the inequalities follow from \eqref{eq:ab}.
    The first is immediate applying Theorem~\ref{thm:matroid-inequality} (b). \\
    For the second, by Theorem~\ref{thm:matroid-inequality} (c), $\mathcal{S}_1(|E|) < 0$. Recalling that $\lambda\chi_{M/e}-\chi_M=\chi_{M/e}+(\lambda-1)\mathcal S_1$ (Proposition~\ref{prop:We}) and $\mathcal S_0=\chi_{M/e}-\mathcal S_1$ by definition, we get
    \[ |E|\chi_{M/e}(|E|)-\chi_M(|E|) = \chi_{M/e}(|E|)+(|E|-1)\mathcal{S}_1(|E|)<\chi_{M/e}(|E|) < \mathcal S_0(|E|).\]
\end{proof}

\section{The port of the evaluation matroid}\label{subsec:rm-codes}

From now on we do not develop further results in the general matroid setting, rather we focus on our matroid 
$\M=(V,\overline r)$. Recall
that $E=V=(\FF_2)^n$ and $N=|V|=2^n$, while $\overline r(\M)=D$. We also assume $2 \le d \le n-1$, throughout the rest of this paper, unless otherwise specified.\\
The values $\sigma_k^{(0)},\sigma_k^{(1)},\mu_k$ are generally defined in term of an element $e \in E$, that is neither a loop nor a coloop and that has no parallel elements. By Lemma~\ref{lem:coloop}, in $\M$ we can choose any $v \in V$ to define these values. Moreover, by translation invariance, $\sigma_k^{(0)},\sigma_k^{(1)},\mu_k$ are constant as functions of $v$. 
\\

First, we specialise Theorem~\ref{thm:matroid-inequality} and Corollary~\ref{cor:inequality} to $\M=(V,\overline r)$. 

\begin{corollary}\label{cor:alpha-inequality}
In $\M$ with $d\le n-1$, for any $v\in V$:
\begin{enumerate}[label=\textnormal{(\alph*)}]
\item $\mathcal S_0(N)\ge 0$;
\item $N\,W_v(N)\ge- \mathcal{M}(N)$;
\end{enumerate}
Consequently
\begin{equation}\label{eq:lower}
\mathcal S_0(N)-\mathcal{M}(N) \le \alpha_1 - \alpha_0.
\end{equation}

\begin{enumerate}[label=\textnormal{(\alph*)}]
\setcounter{enumi}{2}
\item if $\mathcal{P}_v \neq \emptyset$ and every member of $\mathcal{P}_v$ has odd size, then $W_v(N)<0$ and
\end{enumerate}
\begin{equation}\label{eq:upper} \alpha_1 - \alpha_0 < \frac{\alpha_1}{N} < \mathcal S_0(N).\end{equation}
\end{corollary}

\begin{proof}
Apply Theorem~\ref{thm:matroid-inequality} with ground set $E = V$, so $|E|=N$, matroid rank $r(\M)=D$, and $\alpha_0=\chi_{\M}(N)$, $\alpha_1/N=\chi_{\M/v}(N)$.
\end{proof}

To bound the remaining quantities, we connect our matroid construction to coding theory via Reed--Muller codes, which were introduced by Reed and Muller in \cite{reed,RM}. Indeed, the minimal sets of the port turn out to be the
minimal-support codewords of a dual $\RM$ code, so that the arithmetic quantities $\mu_j$ governing the error term become weight-enumerator coefficients.

We recall here some basic definitions and properties, more can be found at \cite{code}. 

\begin{definition}[\cite{reed,RM}]
For $0 \le d \le n$, the \emph{binary Reed--Muller code} of order $d$ and length $N = 2^n$, denoted $\RM(d,n)$, is the subspace of $(\FF_2)^N$ obtained by evaluating all polynomials in $Q_d$ at all points $v \in V = (\FF_2)^n$:
\[
\RM(d,n) := \{ (f(v))_{v \in V}  \mid f \in Q_d \}.
\]
In particular, for quadratic polynomials, $\RM(2,n) = \{ (f(v))_{v \in V} \mid \deg(f) \le 2 \}$.
\end{definition}

\begin{definition}
For a linear code $\mathbf{C} \subseteq (\FF_2)^\nu$, its \emph{dual code} $\mathbf{C}^\perp$ is defined with respect to the standard dot product:
\[
\mathbf{C}^\perp := \{ c \in (\FF_2)^\nu \mid \forall c' \in \mathbf{C}, \quad c\cdot c' = 0 \}.
\]
The \emph{support} of a codeword $c \in (\FF_2)^\nu$ is $\supp(c) := \big\{ i \in \{1,\dots,\nu\} \mid c_i = 1 \big\}$, and its Hamming weight is $\operatorname{w}(c) = |\supp(c)|$. Thus, 
$c$ is a \emph{minimal support codeword} \cite{ashikhminbarg,minimal} if there is no other codeword $c' \neq 0$ such that $\supp(c') \subseteq \supp(c)$.
\end{definition}

\begin{proposition}[\cite{code}] \label{prop:rm-facts}
    For any $0 \le d \le n-1$, \begin{enumerate}
    \item The minimum distance of $\RM(d,n)$ is $2^{n - d}$;
        \item Every codeword has even weight;
        \item The dual of $\RM(d,n)$ is another Reed--Muller code given by:
\[
\RM(d,n)^\perp = \RM(n - d - 1, n).
\]
    \end{enumerate}
        
\end{proposition}

\begin{theorem}\label{thm:circuits}
A set $C \subseteq V$ is a circuit of the evaluation matroid $\M$ if and only if \\
$C = \supp(c)$ for a non-zero codeword $c \in \RM(d,n)^\perp$ of minimal support.
\end{theorem}

\begin{proof}
Let $c = (c_v)_{v \in V} \in (\FF_2)^N$. By definition, $c \in \RM(d,n)^\perp$ if and only if $\sum_{v \in V} c_v f(v) = 0$ for all $f \in Q_d$. Since $Q_d$ is spanned by the monomial basis $\mathrm{Mon}_d$, this condition holds if and only if $\sum_{v \in V} c_v m(v) = 0$ for all $m \in \mathrm{Mon}_d$. Under the evaluation map $\ev(v) = (m(v))_{m \in \mathrm{Mon}_d} \in (\FF_2)^D$, this is equivalent to:
\[
\sum_{v \in V} c_v \ev(v) = \mathbf{0} \quad \text{in } (\FF_2)^D.
\]
Thus $c\in\RM(d,n)^\perp$ if and only if the vectors $\ev(v)$, $v\in\supp(c)$, sum to $\mathbf 0$. Conversely, if $C\subseteq V$ is dependent, there are $a_v\in\FF_2$, not all zero, with $\sum_{v\in C}a_v\ev(v)=\mathbf 0$; over $\FF_2$ the non-zero $a_v$ equal $1$, so $c=\sum_{v\in C}a_v e_v$ is a non-zero word of $\RM(d,n)^\perp$ with $\supp(c)\subseteq C$. Hence $C$ is dependent if and only if it contains the support of a non-zero word of $\RM(d,n)^\perp$.

A circuit $C$ of $\M$ is defined as a minimal dependent set of evaluation vectors. Therefore, $C$ is a circuit of $\M$ if and only if $C = \supp(c)$ for a non-zero codeword $c \in \RM(d,n)^\perp$ whose support is minimal with respect to set inclusion.
\end{proof}

\begin{corollary}\label{cor:minimal-Tv-RM}
A set $S \subseteq V \setminus \{v\}$ is an element of the $v$-port of $\M$ if and only if $S \cup \{v\} = \supp(c)$ for some minimal support codeword $c \in \RM(n - d - 1, n)$.
\end{corollary}

\begin{proof}
This follows directly from Theorem \ref{thm:circuits} and Proposition~\ref{prop:rm-facts}.
\end{proof}

\begin{remark}
    Since $\mu_j := |\{P \in P_v \mid |P| = j\}$, in $\M$
    $$\mu_j = \big|\big\{c \in \RM(n-d-1,n) \,\mid\, c \text{ minimal-support}, \;  \wt(c) = j+1, \; v \in \supp(c)\big\}\big|\,.$$
\end{remark}

Fix $v \in V$ and let $\mathbf{C} = \RM(n-d-1,n)$ with minimum distance $\dist = 2^{d+1}$ and weight distribution denoted by $A_j= A_j(\mathbf{C})$. We use $A^{\min}_j = A_j^{\min}(\mathbf{C})$ for the number of its minimal-support codewords of weight $j$,
so that $A_j^{min} \leq A_j $, with 
$A_\dist^{min} = A_\dist $.

\begin{corollary}\label{cor:port-rm}
Any set $P \in \mathcal{P}_v$ has odd size $|P| \ge 2^{d+1}-1$, and
\[
\mu_j=\frac{(j+1)\,A^{\min}_{j+1}(\mathbf{C})}{N}\ \le\ \frac{(j+1)\,A_{j+1}(\mathbf{C})}{N}\quad.
\]
\end{corollary}

\begin{proof}
The parity and size statements are Proposition~\ref{prop:rm-facts} and Corollary~\ref{cor:minimal-Tv-RM}. The affine group of $V$ acts transitively on $V$, maps $\mathbf{C}$ onto itself and preserves minimality of supports; hence every point lies in the same number of minimal-support words of weight $j+1$, and counting incidences gives $\mu_j N=(j+1)A^{\min}_{j+1}(\mathbf{C})$.
\end{proof}

\begin{remark}
    It follows directly from Corollary ~\ref{cor:port-rm} that $\mu_j = 0$ for $j \le \dist - 2$.
\end{remark}

\begin{remark}
    No closed formula is known for the distribution of minimal-support codewords of $\mathbf C=\RM(n-d-1,n)$ when $d>2$, hence the numbers $\mu_j$ of members of $\mathcal{P}_v$ are not determined exactly.
    A detailed treatment about this can be found in \cite{minimal}.
\end{remark}

\medskip\noindent
Corollary~\ref{cor:port-rm} converts the sum $\mathcal{M}(N)$ of Theorem~\ref{thm:matroid-inequality} (b) into a weighted enumerator of $\mathbf C$. More precisley, we set
\[
\Phi=\Phi(\mathbf{C}):=\sum_{w=\dist}^Nw\,A_w(\mathbf{C})\,N^{-w}.
\]

\begin{lemma}
\label{lem:small}
If $2\leq d\leq n-1$, then:
\begin{enumerate}[label=\textnormal{(\roman*)}]
    \item $\sigma^{(0)}_k=\binom{N-1}k$ for $0\le k\le \dist - 2$.
    \item  $\mathcal{S}_0(N)\ \ge\ \frac{N^{D-1}}{4}+N^{D-2}$.
    \item $\mathcal{M}(N) \le N^{D+1}\Phi$.
\end{enumerate}
\end{lemma}

\begin{proof}
(i) Since circuits have cardinality at least $\dist$, every subset of $V \setminus \{v\}$ of cardinality $k < \dist -1 $ cannot have a $0$-active element, hence it is fixed by $\iota_0$.

\medskip

(ii) It is the sum in the proof of Theorem~\ref{thm:matroid-inequality} (a), truncated at $k=2$ and lower bounded by using $\frac{(N-1)(N-2)(2N+3)}{6} \ge \frac{N^3}{4}$.
\medskip

(iii) By Corollary~\ref{cor:port-rm}, $\mu_j \le \frac{(j+1) A_{j+1}}{N}$. Substituting $\mu_j$ into $\mathcal{M}$ and factoring out $N^{D+1}$ yields:
\[
\mathcal{M}(N) = \sum_{\substack{j=\mathtt{d}-1 \\j \text{ odd}}}^N \frac{(j+1) A_{j+1}}{N} N^{D+1-j} \le N^{D+1} \sum_{w=\dist}^{N} \frac{w A_w}{N^w} = N^{D+1} \Phi,
\]
with inequality given by adding non-negative quantities.
\end{proof}

\section{Bounding the weighted enumerator}\label{sec:rm-bound}

In order to estimate $\mathcal{M}(N)$, we need to bound $\Phi=\sum_w wA_w(\mathbf C)N^{-w}$. Three facts give our estimate, as summarized below.
We recall that $\mathbf{C}=\mathrm{RM}(n-d-1,n)$.

First, we can consider the normalised weight
$Y_c=1-2\operatorname{w}(c)/N\in[-1,1]$ on
the \emph{dual} code
$\mathbf C^\perp=\RM(d,n)$. Then, 
MacWilliams duality turns $\Phi$ into an average over $\mathbf{C}^\perp$ of an explicit function
$f(Y_c)=\kappa e^{\lambda Y_c}(NY_c-1)$ (Lemma~\ref{lem:macwilliams_derivative}).

Second, because $\mathbf C$ has minimum distance $\dist=2^{d+1}$, the first $\dist-1$
moments of $Y$ over $\mathbf C^\perp$ coincide with those over the whole of
$(\FF_2)^N$ (Lemma~\ref{lem:delsarte_orthogonality}),
and the full-space average of $f$ itself vanishes
     (Lemma~\ref{lem:binomial_identity}). Since the Taylor remainder
     is non-negative, the Taylor polynomial of order $\dist-1$
     contributes a non-positive amount, and $\Phi$ is bounded by the
     average of the Taylor remainder. 

Third, that remainder is controlled by a single bound on $f^{(\dist)}$
(Lemma~\ref{lem:global_derivative_bound}) together with an upper bound on $A_\dist(\mathbf{C})$ (Lemma~\ref{lem:min_weight}). The outcome is
$\Phi=O(N^{-2^d})$, which for $d=2$ is all that Section~\ref{sec:mainthm} needs.

We do not search tight bounds, since only the exponent of $N$ matters
downstream, hence we do not try to optimise the presented constants.
\newpage

We start with an elementary numerical bound.
\begin{lemma}
\label{lem:binomial_identity}
For any positive integer $m$,
\[
\sum_{w=0}^m \binom{m}{w} \left(\frac{m-1}{m+1}\right)^w (m - 1 - 2w) = 0
\]
\end{lemma}

\begin{proof}
Let $a = \frac{m-1}{m+1}$. Using $(1+a)^m = \sum \binom{m}{w} a^w$ and $m a (1+a)^{m-1} = \sum w \binom{m}{w} a^w$:
\[
\sum_{w=0}^m \binom{m}{w} a^w (m - 1 - 2w) = (1+a)^{m-1} \left[ (m-1)(1+a) - 2ma \right]
\]
Substituting $1+a = \frac{2m}{m+1}$ gives $(m-1)\left(\frac{2m}{m+1}\right) - 2m\left(\frac{m-1}{m+1}\right) = 0$.
\end{proof}

\begin{lemma}
\label{lem:min_weight}
The number of minimum-weight codewords of $\mathbf{C} = \RM(n-d-1,n)$ satisfies:
\[
A_{\dist} \le \frac{N^{d+2}}{2^{\frac{(d+1)(d+2)}{2}} K_d} \quad \text{where } K_d = \prod_{i=1}^{d+1} (2^i - 1)
\]
\end{lemma}

\begin{proof}
The support of weight-$\dist$ codewords are $(d+1)$-dimensional affine subspaces of $(\FF_2)^n$.\\
The number of $(d+1)$-dimensional linear subspaces is given by the Gaussian binomial coefficient:
\[
\genfrac{[}{]}{0pt}{}{n}{d+1}_2 = \frac{\prod_{i=0}^d (N - 2^i)}{\prod_{i=0}^d (2^{d+1} - 2^i)}
\]
Factoring $2^i$ out of each term in the denominator yields:
\[
\prod_{i=0}^d (2^{d+1} - 2^i) = 2^{\sum_{i=0}^d i} \prod_{i=0}^d (2^{d+1-i} - 1) = 2^{\frac{d(d+1)}{2}} \prod_{j=1}^{d+1} (2^j - 1) = 2^{\frac{d(d+1)}{2}}K_d
\]
Thus, bounding each term in the numerator by $N$, $\genfrac{[}{]}{0pt}{}{n}{d+1}_2 < \frac{N^{d+1}}{2^{\frac{d(d+1)}{2}}K_d}$. \\
Multiplying by the $N/2^{d+1}$ parallel cosets yields the bound.
\end{proof}

\begin{remark}
    Since $\dist=2^{d+1}$, we can equivalently rewrite the bound in the following way: \[
    A_{\dist}\le \left(\frac{N}{\sqrt{\dist}}\right)^{d+2}\frac{1}{K_d}
    \]
\end{remark}
\begin{lemma}
\label{lem:macwilliams_derivative}
Let $\mathbf{C}^\perp = \RM(d, n)$ and  \[ g(w) := \frac{1}{N+1} \left(\frac{N+1}{N}\right)^{N-1} \left(\frac{N-1}{N+1}\right)^{w-1} (N - 1 - 2w). \]Then:
\[
\Phi = \frac{1}{|\mathbf{C}^\perp|} \sum_{c \in \mathbf{C}^\perp} g(\operatorname{w}(c))
\]
\end{lemma}

\medskip
\begin{proof}
Let $W_{\mathbf{C}}(x, y) = \sum_{w=0}^N A_w x^w y^{N-w}$, the weight enumerator polynomial of $\mathbf{C}$. Differentiating gives \[ \Phi = \frac{1}{N} \left.  \frac{\partial W_C}{\partial x} \right|_{(1/N, 1)}.\] Applying MacWilliams duality \cite[Ch.~5]{code} $W_\mathbf{C}(x, y) = \frac{1}{|\mathbf{C}^\perp|} \sum_{c \in \mathbf{C}^\perp} (y+x)^{N-\operatorname{w}(c)} (y-x)^{\operatorname{w}(c)}$ and differentiating term-by-term yields:
\begin{align*}
    \frac{1}{N} \left. \frac{\partial}{\partial x} \left[ (y+x)^{N-w} (y-x)^w \right] \right|_{(1/N, 1)} &= \frac{1}{N} \left(\frac{N+1}{N}\right)^{N-w-1} 
      \left(\frac{N-1}{N}\right)^{w-1} (N - 1 - 2w) =& \\
  &= \frac{1}{N+1} \left(\frac{N+1}{N}\right)^{N-1}
      \left(\frac{N-1}{N+1}\right)^{w-1} (N - 1 - 2w) =& g(w)
\end{align*}
\end{proof}

\begin{lemma}[Delsarte orthogonality \cite{delsarte}]
\label{lem:delsarte_orthogonality}
Let $Y_c = 1 - \frac{2\operatorname{w}(c)}{N}$. For every $0 \le k \le \dist-1$:
\[
\frac{1}{|\mathbf{C}^\perp|} \sum_{c \in \mathbf{C}^\perp} Y_c^k = \frac{1}{2^N} \sum_{u \in (\FF_2)^N} Y_u^k
\]
\end{lemma}

\begin{proof} Notice that $Y_c = \frac{1}{N}\sum_{x=1}^{N}(-1)^{c_x}$.
Expanding $Y_c^k = \frac{1}{N^k} \sum_{x_1, \dots, x_k = 1}^N (-1)^{c \cdot v}$ for $v = 1_{x_1} + \dots + 1_{x_k}$, where $1_{x_i}$ is the vector of weight 1, with non-zero entry in position $x_i$. Averaging over $\mathbf{C}^\perp$ yields the indicator function of $\mathbf{C}$, evaluated on $v$. Namely it is non-zero if and only if $v \in (\mathbf{C}^\perp)^\perp = \mathbf{C}$. Because $v$ has weight at most $k \le \dist - 1 < \dist$, $v \in \mathbf{C} \iff v = \mathbf{0}$. The same computation over $(\FF_2)^N$ gives the indicator of $((\FF_2)^N)^\perp=\{\mathbf 0\}$; so the two averages coincide.
\end{proof}

In order to provide a complete bound on $\Phi$, one needs to consider an analytic version of it, written as a function of $Y_c$. Let \[f_N(Y_c) = f(Y_c) = g(\operatorname{w}(c)) = \kappa e^{\lambda Y_c} (NY_c - 1),\]  where $\lambda = \frac{N}{2} \ln\left(\frac{N+1}{N-1}\right)$, $\kappa = \frac{1}{N+1} \left(\frac{N+1}{N}\right)^{N-1} \left(\frac{N-1}{N+1}\right)^{N/2-1}$. Hence \[\Phi = \frac{1}{|\mathbf{C}^\perp|} \sum_{c \in \mathbf{C}^\perp} f(Y_c)\].
We will study the function $f$ on the interval $[-1,1]$, indeed $-1 \le Y_c \le 1$ for every $c \in \mathbf{C}^\perp$.

\begin{lemma}\label{lem:derivative_positivity}
Let $N \ge 2$. For every integer $k > 1$ and every real $\xi \ge - 1$:
\[
f^{(k)}(\xi) > 0.
\]
\end{lemma}

\begin{proof}
Differentiating $f(\xi)$ $k$ times via the product rule yields the general derivative identity:
\[
f^{(k)}(\xi) = \kappa e^{\lambda \xi} \lambda^{k-1}L_k(\xi), \qquad \text{where } L_k(\xi)= \lambda(N\xi - 1) + k N.
\]
Because all the constants are positive and $L_k(\xi)$ is linear, it suffices to prove that $f^{(k)}(-1) > 0$. \\
Evaluating the expression at $\xi = -1$ yields:
\[
f^{(k)}(-1) = \kappa e^{-\lambda} \lambda^{k-1} \left[ k N - \lambda(N + 1) \right].
\]
Since $\ln\frac{N+1}{N-1} = 2 \sum_{j=0}^\infty \frac{1}{(2j+1)N^{2j+1}}$, we have $\lambda = 1 + \frac{1}{3N^2} + \frac{1}{5N^4} + \dots \le 1 + \frac{1}{3}\sum_{j\ge1}\frac{1}{N^{2j}}$, that gives \begin{equation}\label{eq:lambda}
1<\lambda \le 1 + \frac{1}{3(N^2-1)}.
\end{equation}
Hence, by $k \ge 2$:
\[
k N - \lambda(N+1) \ge N-1-\frac{1}{3(N-1)} > 0.
\]
\end{proof}

From now on, $e$ denotes the Euler number.

\begin{lemma}\label{lem:global_derivative_bound}
Let $N\ge 2$. For every integer $2\le k\le N$ and every $\xi\in[-1,1]$,
\[
  0<f^{(k)}(\xi)<e\,(k+2).
\]
\end{lemma}

\begin{proof}
Put $A=\bigl(1+\tfrac1N\bigr)^{N-1}$, a direct computation from the definitions of $\kappa$ and $\lambda$ gives
\begin{equation*}
  \kappa e^{\lambda}=\frac{A}{N-1}.
\end{equation*}
Moreover, recall by \eqref{eq:lambda} that $1<\lambda\le1+\frac{1}{3(N^2-1)}$. 
\\
By the monotonicity (Lemma~\ref{lem:derivative_positivity}) it is sufficient to prove the bound only for $f^{(k)}(1)$. 

We claim that the map $k\mapsto f^{(k)}(1)/(k+2)$ is increasing. Indeed, let $l_k=L_k(1)=$\\ 
$\lambda(N-1)+kN>0$, so that $f^{(k)}(1)=\frac{A}{N-1}\lambda^{k-1}l_k$. Since
\[
  (k+2)\,l_{k+1}-(k+3)\,l_k = 2N-\lambda(N-1)>0
\]
and $\lambda>1$, we obtain
\[
  \frac{f^{(k+1)}(1)}{k+3}
  =\lambda\cdot\frac{A}{N-1}\lambda^{k-1}\frac{l_{k+1}}{k+3}
  >\frac{A}{N-1}\lambda^{k-1}\frac{l_k}{k+2}
  =\frac{f^{(k)}(1)}{k+2}.
\]
Thus it is sufficient to prove that
\begin{equation*}
  f^{(N)}(1)=\frac{A}{N-1}\lambda^{N-1}\bigl[(N-1)\lambda+N^2\bigr]<e\,(N+2).
\end{equation*}

\medskip
We bound the three factors separately.
\begin{itemize}
  \item By \eqref{eq:lambda}, $(N-1)\lambda\le N-1+\frac1{3(N+1)}$, hence
        $(N-1)\lambda+N^2<N(N+1)$.
  \item Since $(1+\frac1N)^N<e$,
        \[
          \frac{N}{N-1}\,A=\frac{N^2}{N^2-1}\Bigl(1+\frac1N\Bigr)^{N}<e\,\frac{N^2}{N^2-1}.
        \]
  \item By \eqref{eq:lambda} and $e^x\le\frac1{1-x}$ (for $x <1$),
        \[
          \lambda^{N-1}\le e^{(N-1)(\lambda-1)}\le e^{\frac1{3(N+1)}}\le\frac{3(N+1)}{3N+2}.
        \]
\end{itemize}
Combining the three bounds,
\[
  f^{(N)}(1)<(N+1)\,\frac{N}{N-1}A\,\lambda^{N-1}
  <e\,(N+1)\cdot\frac{3N^2}{(N-1)(3N+2)} .
\]
Finally, $\frac{3N^2}{(N-1)(3N+2)}\le\frac{N+2}{N+1}$ is equivalent to $N^2-2N-2\ge0$, which holds for $N\ge3$.
\medskip

\noindent The case $N=2$ is carried out by direct computation. Only $k = 2$ occurs and 
\[
  f^{(2)}(1)=\tfrac32\,\ln3(\ln3+4)<4e. 
\]

\noindent Hence $f^{(N)}(1)<e(N+2)$.
\end{proof}
\newpage

Applying the previous with $k=\dist$ we get the following useful corollary.

\begin{corollary}\label{cor:bound}
    Let $2 \le d \le n-1$. For any $\xi \in [-1,1]$, \[
    0 < f^{(\dist)}(\xi)<e(\dist+2).
    \]
\end{corollary}

\medskip\noindent The three ingredients now combine.

\begin{theorem}
\label{prop:Phi}
Fix $2 \le d \le n-1$. Let $\mathbf{C} = \RM(n-d-1, n)$, $K_d = \prod_{i=1}^{d+1} (2^i - 1)$. Then the weighted sum $\Phi = \sum_{w=\dist}^N w A_w N^{-w}$ satisfies the following inequality:
\[
\Phi \le e(\dist+2)\left( \frac{1}{N^{2^d} \, \dist!!} + \frac{1}{N^{2^{d+1} - d - 2}\,K_d\,\dist^{\frac{d+2}{2}}} \right)
\]
\end{theorem}

\begin{proof}
Expand $f(Y_c)$ in its $(\dist-1)$-th degree Taylor polynomial around $Y = 0$:
\[
f(Y_c) = P_{\dist-1}(Y_c) + \frac{f^{(\dist)}(\xi_q)}{\dist!} Y_c^\dist \quad \text{where } \xi_q \in [-1, 1].
\]

Write $f=P+R$ with $R(Y)=\frac{f^{(\dist)}(\xi_Y)}{\dist!}Y^{\dist}$. By Lemma \ref{lem:delsarte_orthogonality} the averages of $P$ over $\mathbf C^\perp$ and over $(\FF_2)^N$ agree; by Lemma \ref{lem:binomial_identity} the full-space average of $f$ is $0$. Hence
$$
\Phi=\frac{1}{2^N}\sum_{c \in (\FF_2)^N}P(Y_c)+\frac{1}{|\mathbf{C}^\perp|}\sum_{c\in \mathbf C^\perp}R(Y_c)
=\frac{1}{|\mathbf{C}^\perp|}\sum_{c \in \mathbf C^\perp}R(Y_c)-\frac{1}{2^N}\sum_{c \in (\FF_2)^N}R(Y_c)\ \le\ \frac{1}{|\mathbf{C}^\perp|}\sum_{c \in \mathbf C^\perp}R(Y_c),
$$
the last step because $R\ge0$ on $[-1,1]$, by Corollary~\ref{cor:bound}.
Applying the same Corollary~\ref{cor:bound} gives:
\[
\Phi \le \frac{e(\dist+2)}{\dist!} \left( \frac{1}{|\mathbf{C}^\perp|} \sum_{c \in \mathbf{C}^\perp} Y_c^\dist \right)
\]
Expand the $\dist$-th moment into character indicators $v = 1_{x_1} + \dots + 1_{x_{\dist}}$, as in Lemma~\ref{lem:delsarte_orthogonality}:
\[
\frac{1}{|\mathbf{C}^\perp|} \sum_{c \in \mathbf{C}^\perp} Y_c^\dist = \frac{1}{N^\dist} \left| \{ (x_1, \dots, x_{\dist}) \mid v \in \mathbf{C} \} \right|.
\]
Since $\operatorname{w}(v) \le \dist$, $v \in \mathbf{C}$ if and only if $v = \mathbf{0}$ (contributing at most $(\dist-1)!! N^{2^d}$ paired tuples) or $v \in \mathbf{C} \setminus \{\mathbf{0}\}$ of weight $\dist$ (contributing $\dist! A_{\dist}$ tuples). Substituting Lemma~\ref{lem:min_weight} for $A_{\dist}$:
\[
\frac{1}{|\mathbf{C}^\perp|} \sum_{c \in \mathbf{C}^\perp} Y_c^\dist \le \frac{(\dist-1)!! N^{2^d} + \dist! \left( \frac{N}{\sqrt{\dist}}\right)^{d+2}\frac{1}{K_d}}{N^{2^{d+1}}} = \frac{(\dist-1)!!}{N^{2^d}} + \frac{\dist!}{\dist^{\frac{d+2}{2}}\cdot K_d \cdot N^{2^{d+1} - d - 2}}
\]
Multiplying by $\frac{e(\dist+2)}{\dist!}$ completes the proof.
\end{proof}

\begin{corollary}
\label{cor:d2_eval}
For any $2 \le d \le n - 1$, \begin{equation*}
    \Phi < \frac{1}{10 N^4}.
\end{equation*}
\end{corollary}

\begin{proof}
 Let $b_d$ denote the upper bound on $\Phi$ in Theorem~\ref{prop:Phi}, then for any $d \ge 2$, $b_d \le b_2$: indeed, both terms of $b_d$ decrease with $d$, because the exponents $2^d$ and $2^{d+1}-d-2$ increase, while $(\dist+2)/\dist!!$ and $(\dist+2)/(K_d\,\dist^{(d+2)/2})$ decrease. 
 
 \medskip
 \noindent Explicitly computing $b_2$ we have that for $d = 2$, $\dist =  8$, $K_2 = 21$, $2^d = 2^{d+1} - d - 2 = 4$. Hence it evaluates to:
\[
\Phi \le 10e\left( \frac{1}{384\cdot N^4} + \frac{1}{1344\cdot N^4} \right) \le \frac{90e}{2688\cdot N^4} < \frac{1}{10 N^4}
\]
\end{proof}

\section{Proof of the main theorem}\label{sec:mainthm}

We can now prove Theorem~\ref{thm:intro}, restated here as
Theorem~\ref{thm:main}, which is the main result of this section.
The section concludes with a subsection devoted to the case $d=1$, whose results show
that the hypothesis $d\ge2$ cannot be dropped.

\begin{lemma}\label{lem:growth}
Let $n\ge 2$ and $2\le d\le n-1$. Then
\[
  \frac{\alpha_1}{N}\;\ge\;\frac{N^{D-1}}{4}.
\]
\end{lemma}

\begin{proof}
    Note that $\alpha_1/N=\chi_{\M/v}(N)=\mathcal S_0(N)+\mathcal S_1(N)$ by \eqref{eq:ab} and Theorem~\ref{thm:boolean-matroid-link}, and \\ $\mathcal S_1(N)\ge-\frac1N\mathcal{M}(N)\ge-N^{D}\Phi\ge-\frac{N^{D-4}}{10}$ by Theorem~\ref{thm:matroid-inequality} (b), Lemma~\ref{lem:small} and Theorem~\ref{prop:Phi}. Hence, by Lemma~\ref{lem:small},
\begin{equation*}
\frac{\alpha_1}N\ \ge\ N^{D-2}\Bigl(1-\frac1{10N^2}\Bigr)+\frac{N^{D-1}}4 \ge\ \frac{N^{D-1}}4 .
\end{equation*}
\end{proof}

\begin{theorem}\label{thm:main}
Let $n\ge 2$ and $2\le d\le n$. Then
\[
  \Bigl(1-\frac{2}{5N^{2}}\Bigr)\frac{\alpha_1}{N}
  \;\le\;\alpha_1-\alpha_0\;\le\;\frac{\alpha_1}{N},
\]
with equality in the upperbound if and only if $d = n$.\\
In particular, the following inequality holds: \begin{equation}\label{eq:main}
    \alpha_1 \ > \ \alpha_0
\end{equation}
\end{theorem}

\begin{proof}
\emph{Case $d=n$.} By Lemma~\ref{lem:coloop} every element of $V$ is a coloop, hence, thanks to Corollary~\ref{cor:we-edge} (a), we directly get the equality $\alpha_1-\alpha_0 = \frac{\alpha_1}{N}$. In this case, the lower bound of the statement is trivial.  

\medskip
\emph{Case $d\le n-1$.} By Corollary~\ref{cor:port-rm} every member of $\mathcal{P}_v$ has odd cardinality, so Corollary~\ref{cor:alpha-inequality} - \eqref{eq:upper} and \eqref{eq:lower} respectively - gives
\[
\alpha_1-\alpha_0\le\frac{\alpha_1}N\qquad\text{and}\qquad \mathcal S_0(N)\ge\frac{\alpha_1}N .
\]
Applying directly \eqref{eq:lower}, Lemma~\ref{lem:small} and Theorem~\ref{prop:Phi},
\[
\alpha_1-\alpha_0\ \ge\ \mathcal S_0(N)-\mathcal{M}(N)\ \ge\ \mathcal S_0(N)-N^{D+1}\Phi\ \ge\ \frac{\alpha_1}N-\frac{N^{D-3}}{10}.
\]
Hence, by Lemma~\ref{lem:growth},
\begin{equation*}
\frac{N^{D-3}}{10}\le\frac2{5N^2}\cdot\frac{\alpha_1}N,
\end{equation*}
which proves the lower bound.
\medskip
By Corollary~\ref{cor:alpha-inequality} the upper bound is strict for $d \le n-1$.
\end{proof}

\subsection{The case $d=1$}
We now turn to the case $d=1$, excluded from Theorem~\ref{thm:main}. Here
$D=n+1$, and a square system is an affine map $S(x)=Ax+b$ with
$A\in(\FF_2)^{n\times n}$ and $b\in(\FF_2)^n$; the $N^{D}$ systems are the pairs $(A,b)$.
For $0\le r\le n$ let $\rho_r$ be the number of matrices $A\in(\FF_2)^{n\times n}$ of
rank $r$, let $\genfrac{[}{]}{0pt}{}{n}{r}_2$ be the number of $r$-dimensional
subspaces of $(\FF_2)^n$, and for $s\ge0$ put
\[
\pi_s=\prod_{j=1}^{s}\bigl(1-2^{-j}\bigr),\qquad \pi_0=1 .
\]

The following lemma is well known, see e.g.\ \cite{landsberg,kolchin}, but we
provide the proofs for the readers' convenience.

\begin{lemma}\label{lem:affine}
Let $n\ge1$ and $d=1$. Then:
\begin{enumerate}[label=\textnormal{(\roman*)}]
\item $\rho_r=\genfrac{[}{]}{0pt}{}{n}{r}_2\prod_{i=0}^{r-1}(N-2^{i})$ for $0\le r\le n$;
\item $\displaystyle \alpha_1=N\,\rho_n=N\,|\mathrm{GL}(n,\FF_2)|$ and
  $\displaystyle \alpha_0=\sum_{r=0}^{n-1}\rho_r\bigl(N-2^{r}\bigr)$.
\end{enumerate}
\end{lemma}

\begin{proof}
(i) Let $A$ have rank $r$ and column space $U$. Fix a basis $u_1,\dots,u_r$ of
each $r$-dimensional subspace $U$ of $(\FF_2)^n$, and let $B_U$ be the $n\times r$
matrix with columns $u_1,\dots,u_r$. Every column of $A$ is a unique combination
of the $u_i$, so $A=B_UC$ for a unique $C\in(\FF_2)^{r\times n}$; since $B_U$ has
independent columns, $\Rank C=\Rank A=r$, i.e.\ the rows of $C$ are linearly
independent. Conversely, for every $r$-dimensional $U$ and every
$C\in(\FF_2)^{r\times n}$ with independent rows, $B_UC$ has column space $U$ and
rank $r$. Hence the matrices of rank $r$ are in bijection with the pairs $(U,C)$.
There are $\genfrac{[}{]}{0pt}{}{n}{r}_2$ choices of $U$; the $(i+1)$-st row of
$C$ must avoid the span of the previous $i$ rows, which has $2^{i}$ elements, so
there are $\prod_{i=0}^{r-1}(N-2^{i})$ choices of $C$.

(ii) The zero set $S^{-1}(0)$ is empty if $b\notin\operatorname{Im}A$, and a
coset of $\ker A$, of size $2^{\,n-\Rank A}$, if $b\in\operatorname{Im}A$. So $S$
has exactly one zero iff $\Rank A=n$, and then every $b$ occurs; and $S$ has no
zero iff $b$ is one of the $N-2^{r}$ vectors outside $\operatorname{Im}A$, where
$r=\Rank A<n$.
\end{proof}

\begin{theorem}\label{thm:affine}
Let $n\ge3$ and $d=1$. Then:
\begin{enumerate}[label=\textnormal{(\roman*)}]
\item $\displaystyle \frac{\alpha_0}{\alpha_1}
  =\sum_{s=1}^{n}\frac{1-2^{-s}}{2^{s^{2}}\,\pi_s^{2}}
  \prod_{i=0}^{s-1}\Bigl(1-\frac{2^{i}}{N}\Bigr)$;
\item $\displaystyle \frac{\alpha_0}{\alpha_1}\ >\ 1$, equivalently $\alpha_0>\alpha_1$;
\item $\frac{\alpha_0}{\alpha_1}$ is strictly increasing in $n$, and
  \[
  \lim_{n\to\infty}\frac{\alpha_0}{\alpha_1}
  =\sum_{s\ge1}\frac{1-2^{-s}}{2^{s^{2}}\,\pi_s^{2}} \qquad
  \frac{4}{3} < \lim_{n\to\infty}\frac{\alpha_0}{\alpha_1} < 27/20    \]
\end{enumerate}
\end{theorem}

\begin{proof}
(i) Put $s=n-r$. Recall that
$\genfrac{[}{]}{0pt}{}{n}{r}_2=\genfrac{[}{]}{0pt}{}{n}{s}_2
=\prod_{i=0}^{s-1}\frac{N-2^{i}}{2^{s}-2^{i}}$. 
Since
\[
\prod_{i=0}^{s-1}(2^{s}-2^{i})=2^{s^{2}}\pi_s,
\quad
\prod_{i=r}^{n-1}(N-2^{i})=N^{s}\pi_s,
\quad
\prod_{i=0}^{s-1}(N-2^{i})=N^{s}\prod_{i=0}^{s-1}\Bigl(1-\frac{2^{i}}N\Bigr)
\,,
\]
by Lemma~\ref{lem:affine}(i),
\[
\frac{\rho_r}{\rho_n}=\frac{\genfrac{[}{]}{0pt}{}{n}{r}_2}{\prod_{i=r}^{n-1}(N-2^{i})},
\qquad \frac{\rho_r}{\rho_n}=2^{-s^{2}}\pi_s^{-2}\prod_{i<s}(1-\frac{2^{i}}{N}) \,.
\]
Hence  (i) follows
from Lemma~\ref{lem:affine} (ii), since $(N-2^{r})/N=1-2^{-s}$.

(ii) All terms in (i) are positive. Those with $s=1$ and $s=2$ are
$\frac{N-1}{N}$ and $\frac{(N-1)(N-2)}{3N^{2}}$, whose sum is
$\frac{4N^{2}-6N+2}{3N^{2}}$; this exceeds $1$ iff $N^{2}-6N+2>0$, which holds
for $N\ge8$.

(iii) For fixed $s$ the term of (i) increases strictly with $N$ and tends to
$(1-2^{-s})\,2^{-s^{2}}\pi_s^{-2}$; passing from $n$ to $n+1$ also adds the
positive term $s=n+1$. The series converges since $\pi_s\ge\pi_\infty>0$ and
$2^{-s^{2}}$ is summable, and the limit follows by monotone convergence.

As regards the lower/upper bounds,
write $t_s=\frac{1-2^{-s}}{2^{s^2}\pi_s^2}$. Since $\pi_{s+1}=\pi_s(1-2^{-s-1})$,
\[
  \frac{t_{s+1}}{t_s}=\frac{2^{-(2s+1)}}{(1-2^{-s})(1-2^{-s-1})}, \qquad t_{s+1}<4^{-s}t_s<\frac1{10}\,t_s.
\]
Starting from $t_1=1$, this gives $t_2=\frac13$, $t_3=\frac1{63}$ and $t_4=\frac1{6615}$. For $s\ge3$ the denominator is at least $\frac78\cdot\frac{15}{16}>\frac12$, and since $t_{s+1}<\frac1{10}\,t_s$, we have
$\sum_{s\ge4}t_s< t_4\sum_{k\ge 0}\left( \frac{1}{10}\right)^{k} = \frac{10}{9}\,t_4<\frac1{5000}$.
Therefore
\[
  \tfrac43<\tfrac{85}{63}=t_1+t_2+t_3<\lim_{n\to\infty}\tfrac{\alpha_0}{\alpha_1}<\tfrac{85}{63}+\tfrac1{5000}<\tfrac{27}{20}.
\]
\end{proof}

By Theorem~\ref{thm:affine} (ii), the hypothesis $d\ge2$ in Theorem~\ref{thm:main}
cannot be dropped: for $d=1$ statement \eqref{eq:main} is reversed for every $n\ge3$, while the upper bound holds trivially. 

\begin{remark}
For $n=1,2$ the formulae of Lemma~\ref{lem:affine} (ii)
still hold and give $(\alpha_0,\alpha_1)=(1,2)$ and $(21,24)$, so
$\alpha_0<\alpha_1$ there. 
\end{remark}

\section{Remarks on complexity}\label{sec:complexity}

We record the elementary facts used in the introduction. A \emph{quadratic system} over $\FF_2$ is a tuple $F=(f_1,\dots,f_m)$ with $f_i\in Q_2$ in $n$ variables, and $\mathcal{V}(F)\subseteq(\FF_2)^n$ is its set of (binary) solutions. A polynomial-time map $F\mapsto F'$ between systems is \emph{parsimonious} if $|\mathcal{V}(F')|=|\mathcal{V}(F)|$.

\begin{proposition}\label{prop:square}
Let $F$ be a quadratic system over $\FF_2$ in $n$ variables with $m$ equations.
\begin{enumerate}[label=\textnormal{(\alph*)}]
\item If $m\le n$, appending $n-m$ equations $0=0$ gives a square system with
the same solutions.
\item If $m>n$, adjoining $m-n$ variables not occurring in any equation gives a square system whose solution set is $\mathcal{V}(F)\times(\FF_2)^{\,m-n}$; this is square but
not parsimonious, the number of solutions being multiplied by $2^{\,m-n}$.
\item If $m>n$, replacing a pair of equations $f=0$, $g=0$ by
\[
w+f=0,\qquad u+g=0,\qquad w+u+wu=0
\]
in two new variables $w,u$, and iterating $m-n$ times, gives a square system with $2m-n$ variables and equations, with the same \emph{number} of solutions of $F$.
\end{enumerate}
All three maps are computable in polynomial time, (a) and (c) are parsimonious.
\end{proposition}

\begin{proof}
Only (c) needs a word. For each $v \in (\FF_2)^n$ the first two equations force $w=f(v)$ and
$u=g(v)$, and over $\FF_2$ the equation $w+u+wu=0$ holds iff $w=u=0$; so the solutions of the new system are
in bijection with those of the old one, and the equations are in $Q_2$. One step
replaces two equations by three and adds two variables, so it leaves the number of
equations minus the number of variables one lower than before: starting from the
surplus $m-n>0$, after $m-n$ steps the system has $m+(m-n)=2m-n$ equations in
$n+2(m-n)=2m-n$ variables, hence $F'$ is square.
\end{proof}

\begin{proposition}\label{prop:clauses}
There is a parsimonious polynomial-time map from CNF formulae to quadratic systems over $\FF_2$.
\end{proposition}

\begin{proof}
Following and adapting \cite{fraenkelyesha}, we identify a literal $\ell$ in the variables $x_1,\dots,x_n$ with the affine function $x_i$ or $1+x_i$, so that $\ell$ is true iff the function equals $1$. For a clause $C=\ell_1\vee\dots\vee\ell_k$, $k\ge3$, introduce variables $z_1,\dots,z_{k-2}$ and the equations
\[
z_1+\ell_1+\ell_2+\ell_1\ell_2=0,\qquad z_j+z_{j-1}+\ell_{j+1}+z_{j-1}\ell_{j+1}=0\ \ (2\le j\le k-2),\qquad z_{k-2}+\ell_k+z_{k-2}\ell_k+1=0
\]
(for $k=2$ the single equation $\ell_1+\ell_2+\ell_1\ell_2+1=0$; for $k=1$ the equation $\ell_1+1=0$). Since $a\vee b$ corresponds to $a+b+ab$ over $\FF_2$, the first $k-2$ equations force $z_j=\ell_1\vee\dots\vee\ell_{j+1}$ for every assignment of the $x_i$, and the final equation holds iff $C$ is satisfied. The equations are in $Q_2$, and an assignment of the $x_i$ extends to a solution of the system in exactly one way iff it satisfies $C$, in no way otherwise. Doing this for every clause, with disjoint sets of auxiliary variables, gives the map.
\end{proof}

Let $\MQ_k(n)$ be, as in the introduction, the set of square systems with exactly $k$ solutions (in $n$ variables and $n$ equations).\\
We define $\mathbf{MQ} = \bigcup_{k\ge1,n\ge 2}\MQ_k(n)$, $\mathbf{MQ_0} = \bigcup_{n\ge2}\MQ_0(n)$, $\mathbf{MQ_1} = \bigcup_{n\ge2}\MQ_1(n)$.\\
In other words, a system is in $\mathbf{MQ}$ if and only if \\ it has at least a solution, it is in $\mathbf{MQ_1}$ if it has exactly one solution, and it is in $\mathbf{MQ_0}$ if it has no solution at all.
\\
We can implicitely consider them as languages in $\{0,1\}^*$ with the obvious encoding,
where $\mathbf{MQ_0}$ can ben seen as the complement\footnote{as usual in complexity theory, we ignore strings which cannot correspond to an encoded instance} of $\mathbf{MQ}$.

\begin{corollary}\label{cor:complexity} With the notation above, 
\begin{enumerate}[label=\textnormal{(\alph*)}] 
\item $\mathbf{MQ}$ is NP-complete and $\mathbf{MQ_0}$ is coNP-complete.
\item $\mathbf{MQ_1}$ is in DP, it is coNP-hard, and it is NP-hard under randomised polynomial-time reductions.
\end{enumerate}
\end{corollary}

\begin{proof}
Propositions~\ref{prop:clauses} and~\ref{prop:square} compose to a parsimonious polynomial-time reduction $\Omega$,\\$\Omega :\psi\mapsto F_\psi$ from CNF formulae to square quadratic systems: $\psi$ has exactly $s$ satisfying assignments iff $F_\psi$ has exactly $s$ solutions. Therefore, $\Omega$ sends $\mathrm{CNF-SAT}$ (which is NP-complete)\\
to $\mathbf{MQ}$ (which is thus NP-hard) and $\Omega$ sends $\mathrm{CNF-UNSAT}$ (which is coNP-complete) to $\mathbf{MQ_0}$ (which is thus coNP-hard). \\ (a) Membership of $\mathbf{MQ}$ and $\mathbf{MQ_0}$ in NP and coNP, respectively, is clear. \\ (b) A square system has exactly one solution iff it has at least one (an NP property) and at most one (a coNP property: two distinct solutions are a certificate of failure); so the language is in DP \cite{PAPADIMITRIOU1984244}. The reduction $\Omega$ maps the unique-satisfiability language UniqueSAT\footnote{UniqueSAT is the languages of Boolean formulae with exactly one satisfying assignment; we are considering this notion and not its promise version.} into $\mathbf{MQ_1}$ and its complement into the complement, so it is a many-one reduction from unique satisfiability, which is coNP-hard \cite{blassgurevich} and NP-hard under randomised reductions \cite{valiantvazirani},
thus $\mathbf{MQ_1}$ inherits both hardness properties.
\end{proof}

\section{Open problems}
\label{sec:open}

The first natural question that comes to mind after seeing Theorem \ref{thm:main} is whether there is anything special about
the binary field.
Any attempt to generalise will start from the following experimental observations.

Fix a prime power $q$, let $V=(\FF_q)^{\,n}$, $N=q^n$, and let $Q_{q,d}$ be the space
of functions $V\to\FF_q$ of total degree at most $d$, written in reduced form with
all exponents at most $q-1$; put $D_q=\dim Q_{q,d}$ and let $\alpha_k(q,n,d)$ be the
number of square systems of $n$ such functions having exactly $k$ zeros in $V$, out
of $q^{\,nD_q}$. 
Exact enumeration of all $q^{\,2D_q}$ systems at $n=2$, $d=2$ gives
\[
\begin{array}{r|r|r|r|r}
q & N & \alpha_0 & \alpha_1 & \alpha_1/\alpha_0\\\hline
2 & 4  & 81            & 108           & 1.3333\\
3 & 9  & 184\,184      & 206\,712      & 1.1223\\
4 & 16 & 6\,012\,645   & 6\,242\,400   & 1.0382\\
5 & 25 & 88\,849\,944  & 88\,998\,000  & 1.0017\\
7 & 49 & 5\,102\,401\,296 & 4\,956\,042\,672 & 0.9713
\end{array}
\]
so that
\[
\alpha_0(7,2,2)-\alpha_1(7,2,2)=146\,358\,624>0 ,
\]
and thus \eqref{eq:main} fails over $\FF_7$ (and we think it is more evident with larger fields).\\  
The remaining possibilities are 
$$
\boxed{q=3,\qquad q=4, \qquad  q=5},
$$
which we present as open problems.\\

\begin{remark}
We argue now the difficulty of a trivial transposition of our proof to other finite fields, looking at two proof technicalities.\\    
First, note that the evaluation code of $Q_{q,d}$ is the generalised Reed--Muller code
$\GRM_q(d,n)$, whose dual is $\GRM_q\bigl(n(q-1)-d-1,\,n\bigr)$.  At $d=2$ the
minimum weight of that dual is
$
8 \ (q=2),\qquad 6 \ (q=3),\qquad 4 \ (q\ge4).
$
For $q\ge4$ a minimum-weight word is carried by four collinear points: the
restrictions to a line of the functions of degree at most two form a
$3$-dimensional space, so any four points of a line carry a dependence among their
evaluation functionals, and every such dependence extends to $Q_{q,2}$ because the
restriction map onto the line is onto.  Over $\FF_2$ a line has two points and the
configuration does not exist; this is the source of the error term $O(N^{-2^d})$ of
Section~\ref{sec:rm-bound}, and it is needed in our proof of Theorem~\ref{thm:main}.\\
Second, 
the sign-reversing
involution of Section~\ref{sec:reduction} requires every member of the port to
have odd size, which follows from all words of $\RM(n-d-1,n)$ having even weight. Instead, no parity constraint on minimal supports survives over $\FF_q$.
\end{remark}

By Theorem~\ref{thm:main}, for every $n\ge2$ there exist injections $\MQ_0(n)\hookrightarrow\MQ_1(n)$ whose images omit at most a $2^{-n}$ fraction of $\MQ_1(n)$. While this counting argument does not yield an explicit map, a partial injection $\tau$ is developed in \cite{radicisala}. 
\\
We propose as an open problem finding a 
\begin{center}
total and efficiently computable injection 
\end{center}
between these sets.

\section*{Acknowledgements}

These results are included in the first author's MSc thesis, supervised by the
second author. 
Previously, some of these have been presented at the Italian conference CIFRIS26 (Rome, Italy) and at the 
\emph{Security and Systems seminars} in EPITA (Paris, France).
For the interesting discussions on this topic, the authors would like to thank Edoardo Ballico (Trento), Ludovic Perret (Paris) and Marco Pedicini (Rome).

Several of the numerical claims were checked by computer algebra
with the assistance of Anthropic Claude and OpenAI ChatGPT, which also reviewed the
arguments and assisted with some proofs.

\bibliographystyle{plain}
\bibliography{biblio}

\end{document}